\documentclass{article}

\usepackage{amsthm}
\usepackage{fullpage}
\usepackage{amsmath}
\usepackage{algorithm}
\usepackage{algpseudocode}
\usepackage{graphicx}
\usepackage{amssymb}
\usepackage{enumerate}
\usepackage{stmaryrd}

\usepackage{authblk}

\usepackage{color,cite,url,ifthen}

\newcommand{\showComments}{yes}
\newcommand{\submit}{no}            
\newcommand{\myNote}[2]{            
  \ifthenelse{\equal{\submit}{yes}}{}{%
    \ifthenelse{\equal{\showComments}{yes}}{\textcolor{#1}{#2}}{}
  }
}

\newcommand{\eat}[1]{}         
\usepackage{empheq}
\definecolor{myblue}{rgb}{.8, .8, 1}

\newlength\mytemplen
\newsavebox\mytempbox

\makeatletter
\newcommand\mybluebox{%
    \@ifnextchar[
       {\@mybluebox}%
       {\@mybluebox[0pt]}}

\def\@mybluebox[#1]{%
    \@ifnextchar[
       {\@@mybluebox[#1]}%
       {\@@mybluebox[#1][0pt]}}

\def\@@mybluebox[#1][#2]#3{
    \sbox\mytempbox{#3}%
    \mytemplen\ht\mytempbox
    \advance\mytemplen #1\relax
    \ht\mytempbox\mytemplen
    \mytemplen\dp\mytempbox
    \advance\mytemplen #2\relax
    \dp\mytempbox\mytemplen
    \colorbox{myblue}{\hspace{1em}\usebox{\mytempbox}\hspace{1em}}}
\makeatother

\newcommand{\LRU}{{\sc {\bf L}ru}}
\newcommand{\LFU}{{\sc {\bf L}fu}}
\newcommand{\LRFU}{{\sc {\bf L}rfu}}
\newcommand{\LIRS}{{\sc {\bf L}irs}}
\newcommand{\ARC}{{\sc {\bf A}rc}}
\newcommand{\CAR}{{\sc {\bf C}ar}}
\newcommand{\OPT}{{\sc {\bf O}pt}}
\newcommand{\CLOCK}{{\sc {\bf C}lock}}
\newcommand{\QQ}{{\sc {\bf 2}q}}
\newcommand{\Replace}{{\sc Replace}}

\definecolor{shadecolor}{gray}{0.875} 

\newtheorem{theorem}{Theorem}
\newtheorem{lemma}{Lemma}

\title{Analyzing Adaptive Cache Replacement Strategies}
\author[2]{Mario E. Consuegra}
\author[1]{Wendy A. Martinez}
\author[1]{Giri Narasimhan}
\author[1]{Raju Rangaswami}
\author[1]{Leo Shao}
\author[1]{Giuseppe Vietri}
\affil[1]{School of Computing and Information Sciences, 
Florida International University, Miami, FL 33199, USA.
\texttt{\{walem001,giri,raju,gviet001\}@fiu.edu}}
\affil[2]{Google Inc., Kirkland, WA, USA.}

\begin{document}

\maketitle

\begin{abstract}
%
%
  {\em Adaptive Replacement Cache} (\ARC) and {\em CLOCK with Adaptive
    Replacement} (\CAR) are state-of-the-art ``adaptive'' cache
  replacement algorithms invented to improve on the shortcomings of
  classical cache replacement policies such as \LRU, \LFU\ and
  \CLOCK. 
  By separating out items that have been accessed only once and items
  that have been accessed more frequently, both \ARC\ and \CAR\ are
  able to control the harmful effect of single-access items flooding
  the cache and pushing out more frequently accessed items.
  Both \ARC\ and \CAR\ have been shown to outperform their
  classical and popular counterparts in practice. 
  Both algorithms are complex, yet popular. 
  Even though they can be treated as online algorithms with an
  ``adaptive'' twist, a theoretical proof of the competitiveness of
  \ARC\ and \CAR\ remained unsolved for over a decade.
  We show that the competitiveness ratio of \CAR\ (and \ARC) has a
  lower bound of $N+1$ (where $N$ is the size of the cache) and an
  upper bound of $18N$ ($4N$ for \ARC).
  If the size of cache offered to \ARC\ or \CAR\ is larger than the
  one provided to \OPT, then we show improved competitiveness ratios.
  The important implication of the above results are that no
  ``pathological'' worst-case request sequences exist that could
  deteriorate the performance of \ARC\ and \CAR\ by more than a
  constant factor as compared to \LRU.
\end{abstract}

\section{Introduction}

Megiddo and Modha \cite{mm-ARC-03,mm-ARC-04} engineered an amazing
cache replacement algorithm that was self-tuning and called it {\em
  Adaptive Replacement Cache} or \ARC. Later, Bansal and Modha
\cite{bm-CAR-04} designed another algorithm called {\em \CLOCK\ with
  Adaptive Replacement} (\CAR). Extensive experimentation suggested
that \ARC\ and \CAR\ showed substantial improvements over previously
known cache replacement algorithms, including the well-known {\em
  Least Recently Used} or \LRU\ and \CLOCK. On the theoretical side,
the seminal work of Sleator and Tarjan \cite{st-aelupr-85} showed that
\LRU\ can be analyzed using the theory of online algorithms.
They showed that \LRU\ has a competitiveness ratio of $N$ (where $N$
is the size of the cache). More surprisingly, they also showed that
with no prefetching, no online algorithm for cache replacement could
achieve a competitiveness ratio less than $N$, suggesting that under
this measure, \LRU\ is optimal.
In other words, there exist worst-case request sequences that would
prevent any algorithm from being better than $N$-competitive. 
While these results are significant, they highlight the difference
between theory and practice. Sleator and Tarjan's techniques analyze
online algorithms in terms of their {\em worst-case} behavior (i.e.,
over all possible inputs), which means that other algorithms with
poorer competitiveness ratios could perform better in practice.
Another way to state this is that the results assume an {\em oblivious
  adversary} who designs inputs for the online algorithms in a way
that make them perform as poorly as possible. The upper bound on
performance ratio merely guarantees that no surprises are in store,
i.e., there is no input designed by an adversary that can make the
algorithm perform poorly.

Given a fixed size cache, the {\bf cache replacement problem} is that
of deciding which data item to evict from the cache in order to make
room for a newly requested data item with the objective of maximizing
cache hits in the future. The cache replacement problem has been
referred to as a {\em fundamental and practically important online
  problem in computer science} (see Irani and Karlin \cite{h-aanp-97},
Chapter 13) and a ``fundamental metaphor in modern computing''
\cite{mm-ARC-04}. 
%

The \LRU\ algorithm was considered the most optimal page replacement
policy for a long time, but it had the drawback of not being
``scan-resistant'', i.e., items used only once could pollute the cache
and diminish its performance. Furthermore, \LRU\ is difficult to
implement efficiently, since moving an accessed item to the front of
the queue is an expensive operation, first requiring locating the
item, and then requiring data moves that could lead to unacceptable
cache contention if it is to be implemented consistently and
correctly. The \CLOCK\ algorithm was invented by Frank Corbat\'o in
1968 as an efficient one-bit approximation to \LRU\ with minimum
overhead \cite{c-pems-68} and continues to be used in MVS, Unix,
Linux, and Windows operating systems \cite{f-wntprp-99}.  Like \LRU,
\CLOCK\ is also not scan-resistant because it puts too much emphasis
on ``recency'' of access and pays no attention to ``frequency'' of
access.  So there are sequences in which many other algorithms can
have significantly less cost than the theoretically optimal
\LRU. Since then, many other cache replacement strategies have been
developed and have been showed to be better than \LRU\ in
practice. These are discussed below in Section \ref{sec:bg}.

An important development in this area was the invention of {\bf
  adaptive algorithms}. While regular ``online'' algorithms are
usually designed to respond to input requests in an optimal manner,
these {\em self-tuning} algorithms are capable of adapting to changes
in the request sequence caused by changes in the workloads. Megiddo
and Modha's \ARC\ \cite{mm-ARC-03} is a self-tuning algorithm that is
a hybrid of \LFU\ and \LRU. Bansal and Modha's \CAR\ is an
adaptivehybrid of \LFU\ and \CLOCK\ \cite{bm-CAR-04}. Experiments show
that \ARC\ and \CAR\ outperform \LRU\ and \CLOCK\ for many benchmark
data sets \cite{bm-CAR-04}. Versions of \ARC\ have been deployed in
commercial systems such as the IBM DS6000/DS8000, Sun Microsystems's
ZFS, and in PostgreSQL.

Unfortunately, no {\bf theoretical analysis} of the adaptive
algorithms, \ARC\ and \CAR, exist in the literature. The main open
question that remained unanswered was whether or not there exist
``pathological'' request sequences that could force \ARC\ or \CAR\ to
perform poorly.
In this document we show that these two algorithms are
$O(N)$-competitive, suggesting that they are not much worse than the
optimal \LRU. We also prove a surprising lower bound on the
competitiveness that is larger than $N$.

The main contributions of this paper are as follows:
\begin{enumerate}
\item For completeness, we provide proofs that \LRU\ and \CLOCK\ are
  $N$-competitive.
\item We prove a lower bound on the competitiveness of \ARC\ and \CAR\
  of $N+1$, proving that there are request sequences where they cannot
  outperform \LRU\ and \CLOCK.
\item We show that \ARC\ is $4N$-competitive.
\item We show that \CAR\ is $18N$-competitive.
\item We obtain precise bounds for the competitiveness of all four
  algorithms when the sizes of the caches maintained by them are
  different from that maintained by \OPT.
\item We show that if the size of the cache is twice that of the one
  allowed for the optimal offline algorithm, then the competitiveness
  ratio drops to a small constant. 
\end{enumerate}
%
%
We use the method of potential functions to analyze the
algorithms. However, the main challenges in solving these problems is
that of carefully designing the potential function for the
analysis. We discuss the role of the adaptive parameter in the
potential function. The contributions of this paper are summarized in
Table \ref{tab:res}. In this table, $N$ is the size of the cache
maintained by the algorithm, while $N_O$ is the size of the cache
maintained by \OPT. The table provides lower bounds (LB) and upper
bounds (UB) on the competitiveness ratio when the cache sizes are
equal, i.e., $N = N_O$; it also provides upper bounds when they are
not equal.

\begin{table}[htb]
  \centering
  \begin{tabular}[htb]{|c|c|c|c|c|c|} 
    \hline
  Algorithm & Compet. Ratio & Compet. Ratio & Compet. Ratio UB & [Ref] \\
   & LB & UB & w/ Unequal Sizes & \\
  \hline\hline
  \LRU & $N$ & $N$ & $N/(N-N_O+1)$ & \cite{st-aelupr-85} \\ \hline
  \ARC & $N+1$ & $4N$ & $12N/(N-N_O+1)$ & This paper \\ \hline
  \CLOCK & $N$ & $N$ & $N/(N-N_O+1)$ & This paper\\ \hline
  \CAR & $N+1$ & $18N$ & $18N/(N-N_O+1)$ & This paper \\ \hline
\end{tabular}

\caption{Summary of Results}
\label{tab:res}
\end{table}
After providing relevant background on cache replacement algorithms in
Section \ref{sec:bg}, we discuss the lower bounds on the
competitiveness ratios of \ARC\ and \CAR\ in Section
\ref{sec:lb}. Next we prove upper bounds on the competitiveness ratios
of \LRU, \CLOCK, \ARC, and \CAR\ in Section \ref{sec:ub}. Concluding
remarks can be found in Section \ref{sec:concl}.

\section{Previous Work on Cache Replacement Strategies}
\label{sec:bg}

Below we give brief descriptions of the four algorithms being
discussed in this paper, after which we mention a large collection of
other closely related cache replacement algorithms.

\paragraph*{The \LRU\ Algorithm: }
\LRU\ evicts the least recently used entry. It tends to perform well
when there are many items that are requested more than once in a
relatively short period of time, and performs poorly on
``scans''. \LRU\ is expensive to implement because it requires a queue
with move-to-front operations whenever a page is requested. 
%

\paragraph*{The \CLOCK\ Algorithm: }
On the other hand, \CLOCK\ was designed as an efficient approximation
of \LRU, which it achieves by avoiding the move-to-front
operation. \CLOCK's cache is organized as a single ``circular'' list,
instead of a queue. The algorithm maintains a pointer to the ``head''
of the list. The item immediately counterclockwise to it is the
``tail'' of the list. Each item is associated with a ``mark''
bit. Some of the pages in the cache are marked, and the rest are
unmarked. When a page hit occurs that page is marked, but the contents
of the cache remain unchanged. When a page fault occurs, in order to
make room for the requested page, the head page is evicted if the page
is unmarked. If the head page is marked, the page is unmarked and the
head is moved forward clockwise, making the previous head as the tail
of the list. After a page is evicted, the requested page is unmarked
and placed at the tail of the list. \CLOCK\ is inexpensive to
implement, but is not scan-resistant like \LRU. 

\paragraph*{The  \ARC\ Algorithm}
To facilitate our discussion, we briefly describe the \ARC\
algorithm. As mentioned before, it combines ideas of recency and
frequency. \ARC's cache is organized into a ``main'' part (of size
$N$) and a ``history'' part (of size $N$). The main part is further
divided into two lists, $T_1$ and $T_2$, both maintained as LRU lists
(i.e., sorted by ``recency''). $T_1$ focuses on ``recency'' because it
contains pages with short-term utility. Consequently, when an item is
accessed for the first time from the disk, it is brought into
$T_1$. Items ``graduate'' to $T_2$ when they are accessed more than
once. Thus, $T_2$ deals with ``frequency'' and stores items with
potential for long-term utility. Additionally, \ARC\ maintains a
history of $N$ more items, consisting of $B_1$, i.e., items that have
been recently deleted from $T_1$, and $B_2$, i.e., items that have
been recently deleted from $T_2$. History lists are also organized in
the order of recency of access. The unique feature of \ARC\ is its
self-tuning capability, which makes it scan-resistant. Based on a
self-tuning parameter, $p$, the size of $T_1$ may grow or shrink
relative to the size of $T_2$. The details of the algorithm are fairly
complex and non-intuitive. Detailed pseudocode for \ARC\ (Figure 4
from \cite{mm-ARC-03}) is provided in the Appendix for convenience.

It is worth noting that \ARC\ is considered a ``universal'' algorithm
in the sense that it does not use any {\em a priori} knowledge of its
input, nor does it do any offline tuning. Furthermore, \ARC\ is
continuously adapting, since adaptation can potentially happen at
every step.

It must be noted that our results on \ARC\ assume the ``learning
rate'', $\delta$, to be equal to 1, while the \ARC\ algorithm as
presented by Megiddo and Modha recommended a ``faster'' learning rate
based on experiments on real data. The learning rate is the rate at
which the adaptive parameter $p$ is changed as and when needed.

\paragraph*{The  \CAR\ Algorithm}
%
Inspired by \ARC, \CAR's cache is organized into two main lists, $T_1$
and $T_2$, and two history lists, $B_1$ and $B_2$. Inspired by \CLOCK,
both $T_1$ and $T_2$ are organized as ``circular'' lists, with each
item associated with a mark bit. The history lists, $B_1$ and $B_2$
are maintained as simple FIFO lists. We let $t_1, t_2, b_1, b_2$
denote the sizes of $T_1, T_2, B_1, B_2$, respectively. Also, let
$t := t_1 + t_2$. Let lists $L_1$ (and $L_2$, resp.) be the list of
size $\ell_1$ ($\ell_2$, resp.) obtained by concatenating list $B_1$
to the end of ``linearized'' $T_1$ (concatenating $B_2$ to the tail of
$T_2$, resp.). Note that circular lists are linearized from head to
tail. We let $T_1^0$ and $T_2^0$ ($T_1^1$ and $T_2^1$, resp.)  denote
the sequence of unmarked (marked, resp.) pages in $T_1$ and $T_2$,
respectively.

The following invariants are maintained by \CAR\ for the lists: 
\begin{enumerate}
\item $0 \leq t_1 + t_2 \leq N$
\item $0 \leq \ell_1 = t_1 + b_1 \leq N$
\item $0 \leq \ell_1 + \ell_2 = t_1 + t_2 + b_1 + b_2 \leq 2N$
\item $t_1 + t_2 < N \implies b_1 + b_2 = 0$
\item $t_1 + t_2 + b_1 + b_2 \geq N \implies t_1 + t_2 = N$
\item Once  $t_1 + t_2 = N$ and/or $\ell_1 + \ell_2 = 2N$, they remain true
from that point onwards.
\end{enumerate}

\CAR\ maintains an adaptive parameter $p$, which it uses as a target
for $t_1$, the size of list $T_1$. Consequently, $N-p$ is the target
for $t_2$. Using this guiding principle, it decides whether to evict
an item from $T_1$ or $T_2$ in the event that a miss requires one of
the pages to be replaced. The replacement policy can be summarized
into two main points: 


\begin{enumerate} 
\item If the number of items in $T_1$ (barring the marked items at the
  head of the list) exceeds the target size, $p$, then evict an
  unmarked page from $T_1$, else evict an unmarked page from $T_2$.

\item If $\ell_1 = t_1 + b_1 = N$, then evict a history page from $B_1$,
  else evict a history page from $B_2$. Since the details of the
  algorithm are complex, the actual pseudocode is provided (Figure 2
  from \cite{bm-CAR-04}) in the Appendix.
\end{enumerate}

\paragraph*{Other Cache Replacement Algorithms}
The {\sc {\bf D}uelingClock} algorithm \cite{jipp-dcacr-10} is like
\CLOCK\ but keeps the clock hand at the newest page rather than the
oldest one, which allows it to be scan-resistant. More recent
algorithms try to improve over \LRU\ by implementing multiple cache
levels and leveraging history.  In \cite{oow-LRUK-93} the \LRU-$K$
algorithm was introduced. Briefly, the \LRU-$K$ algorithm estimates
interarrival times from observed requests, and favors retaining pages
with shorter interarrival times. Experiments have shown \LRU-$2$
performs better than \LRU, and that \LRU-$K$ does not show increase in
performance over \LRU-$2$ \cite{oow-LRUK-93}, but has a higher
implementation overhead. It was also argued that \LRU-$K$ is optimal
under the independence reference model (IRM) among all algorithms $A$
that have limited knowledge of the $K$ most recent references to a
page and no knowledge of the future \cite{oow-LRUK-93}.

In essence, the \LRU-$K$ algorithm tries to efficiently approximate
{\em Least Frequently Used} (\LFU) cache replacement algorithm. As $K$
becomes larger, it gets closer and closer to \LFU.
It has been argued that \LFU\ cannot adapt well to changing workloads
because it may replace currently ``hot'' blocks instead of ``cold''
blocks that had been ``hot'' in the past. \LFU\ is implemented as a
heap and takes $O(\log N)$ time per request.

Another cache replacement algorithm is \LIRS\ \cite{jz-LIRS-02}.  The
\LIRS\ algorithm evicts the page with the largest IRR (inter-reference
recency). It attempts to keep a small ($\approx 1\%$) portion of the
cache for HIR (high inter-reference) pages, and a large ($\approx
99\%$) portion of the cache for LIR (low inter-reference) pages.  The
\CLOCK-{\sc Pro} algorithm approximates \LIRS\ efficiently using
\CLOCK\ \cite{jcz-cpro-05}.
The \QQ\ \cite{JS94} algorithm is scan-resistant. It keeps a FIFO
buffer $A_1$ of pages that have been accessed once and a main \LRU\
buffer $A_m$ of pages accessed more than once. \QQ\ admits only hot
pages to the main buffer.  The buffer $A_1$ is divided into a main
component that keeps the pages in $A_1$ that still reside in cache,
and a history component that remembers pages that have been evicted
after one access.  The relative sizes of the main and history
components are tunable parameters. \QQ\ has time complexity of $O(1)$.
Another algorithm that tries to bridge the gap between recency and
frequency is \LRFU \cite{lcknmck-LRFU-01}. This is a hybrid of
\LRU\ and \LFU\ and is adaptive to changes in workload. The time
complexity ranges from $O(1)$ for \LRU\ to $O(\log n)$ for \LFU.

\section{Lower Bounds on Competitiveness Ratio for \ARC\ and \CAR} 
\label{sec:lb}

This section presents our results on the lower bounds for \ARC\ and
\CAR. We also show that the adaptive parameter is critical to both
\ARC\ and \CAR\ by showing that their non-adaptive versions have an
unbounded competitiveness ratio. 

\subsection{Lower Bound for \ARC}
\label{sec:lbarc}

First, we show a lower bound on the competitiveness ratio for \ARC.

\begin{theorem}
  \label{thm:arc1}
  The competitiveness ratio of Algorithm \ARC\ has a lower bound of
  $N+1$.
\end{theorem}

\begin{proof}
  We show that we can generate an unbounded request sequence that
  causes $N+1$ page faults on \ARC \: for every page fault on \OPT.
  The sequence only involves $2N+1$ pages denoted by $1, \ldots, 2N+1$.
  Our example, will take the contents of the cache managed by \ARC\
  from configurations 1 through configuration 5, which are shown in
  Table \ref{tab:lbarc}. Note that configuration 1 and configuration 5
  are essentially the same to the extent that the value of $p$ is 0 in
  both, and the number of pages in each of the four parts of the cache
  are identical. 

  \begin{table}[bht]
    \centering
    \begin{tabular}[b]{|c|c|c|c|c|c|}
      \hline
      Configuration & $p$ & $T_1$ & $T_2$ & $B_1$ & $B_2$ \\ 
      \hline \hline
      1 & 0 & $\emptyset$ & $1, \ldots, N$ & $\emptyset$ & $N+1,
                  \ldots, 2N$ \\ \hline
      2 & 0 & $2N+1$ & $2, \ldots, N$ & $\emptyset$ & $N+2,
                  \ldots, 2N, 1$ \\ \hline
      3 & 0 & $\emptyset$ & $2, \ldots, N, 1$ & $2N+1$ & $N+2,
                  \ldots, 2N$ \\ \hline
      4 & 1 & $\emptyset$ & $3, \ldots, N, 1, 2N+1$ & $\emptyset$ & $N+2,
                  \ldots, 2N, 2$ \\ \hline
      5 & 0 & $\emptyset$ & $1, 2N+1, 2, \ldots, N-1$ & $\emptyset$ & $N+2,
                  \ldots, 2N, N$ \\ \hline
    \end{tabular}
    \caption{Example for Lower Bound on \ARC's competitiveness}
    \label{tab:lbarc}
  \end{table}
  We note that we can obtain configuration 1 from an empty cache with
  the following request sequence:
  $2N, 2N, 2N-1, 2N-1, \ldots, 2, 2, 1, 1$.
  Consider the first half of the above request sequence, which
  contains a total of $4N$ requests to $2N$ new pages, each page
  requested twice in succession. The first time a page is requested
  from the first $N$ new pages, it will be put in $T_1$. The second
  time the page is requested, it will get moved to $T_2$. In the
  second half, if a page not in \ARC \: is requested, {\sc Replace}
  will be called, which will move a page from $T_2$ to $B_2$, and the
  new page will be placed in $T_1$. When the same page is requested
  again, it simply gets moved to $T_2$. The value of $p$ remains
  unchanged in this process. It is clear that we get Configuration 1
  as a result of the request sequence.

  We design our request sequence by following the steps below.

\begin{enumerate}
\item Make one request to a page $2N+1$ not in \ARC. We will assume
  that this is a brand new page and therefore also causes a fault for
  \OPT\ and for \ARC. The page $2N+1$ will go into $T_1$ and a page in
  $T_2$ will be demoted to $B_2$. The contents of \ARC \: is given by
  Configuration 2 in Table \ref{tab:lbarc}.
\item Request any page in $B_2$. This decreases the value of $p$ but
  since $p$ is zero it will remain unchanged. Since the size of $T_1$
  is more than $p$ \ARC \: will call {\sc Replace}, which will act on
  $T_1$, hence $2N+1$ will be demoted to $B_1$.  Upon requesting page
  1 in $B_2$, we get Configuration 3 in Table \ref{tab:lbarc}.
\item The next step is to request $2N+1$ again, which will move to
  $T_2$, $p$ is increased and a page in $T_2$ is demoted to
  $B_2$. Configuration 4 reflects the contents of the cache at this
  stage.
\item Finally we make $N-2$ requests to any pages from $B_2$. By
  requesting the pages $2,3,\ldots,N$, we end up in Configuration 5
  from Table \ref{tab:lbarc}.
\end{enumerate}

The steps outlined above cause $N+1$ page faults for \ARC \: and at
most one page fault for \OPT. Since we are back to the initial
configuration we can repeat this process over again. This concludes
the proof that the competitiveness ratio of \ARC \: is at least $N+1$.
\end{proof}

\subsection{Lower Bound for \CAR}

Now we prove a similar lower bound for \CAR.

\begin{theorem}
  \label{thm:carlb}
  The competitiveness ratio of Algorithm \CAR\ has a lower bound of
  $N+1$.
\end{theorem}

\begin{proof}
  We show that we can generate an infinite request sequence that
  causes $N+1$ page faults in \CAR\ for every page fault on \OPT.  The
  sequence only involves $2N+1$ pages denoted by $1, \ldots, 2N+1$.
  Our example, will take the contents of the cache managed by \CAR\
  from configurations 1 through $N+2$ as shown in Table \ref{tab:lbcar}.
  Note that a superscript of 1 on any page in $T_1 \cup T_2$ indicates
  that it is marked. All others are unmarked. Also note that
  configuration 1 and configuration $N+2$ are essentially the same
  upon relabeling.

  First, we show that configuration 1 is attainable, by showing that
  it can be reached from an empty cache. This is formalized in the
  following lemma. 
  \begin{lemma}
    We can obtain configuration 1 starting from an empty cache with the
    following request sequence:
    $2N, 2N, 2N-1, 2N-1, \ldots, 2, 2, 1, 1$. 
    \label{lem:carlb}
  \end{lemma}
  \begin{proof}
    The first half of the above request sequence calls each of the $N$
    pages $2N, 2N-1, \ldots, N+1$ twice in succession. The first time
    they are called, they are moved into $T_1$ unmarked. The second
    time the same page is called it gets marked, but remains in
    $T_1$. At the end of the first half, all the $N$ pages requested
    end up in $T_1$ and are all marked.

    The next call to new page $N$, will trigger a call to {\sc
      Replace}, which will move all the marked pages in $T_1$ to $T_2$
    leaving them unmarked. It will also move one page from $T_2$ to
    $B_2$. Finally, the requested page $N$ will be moved to $T_1$ and
    left unmarked. When requested again, it simply gets marked. When
    the next page, i.e., $N-1$ is requested, it moves marked page $N$
    to $T_2$, moves one more page from $T_2$ to $B_2$. As the rest of
    the pages from the request sequences are requested, the previous
    requested page gets moved to $T_2$, which in turn demotes one of
    its pages to $B_2$. At the end of the process, we get a marked
    page $1$ in $T_1$. Pages $2, \ldots, N$ are in $T_2$, unmarked,
    and pages $N+1, \ldots, 2N$ end up in $B_2$. This is exactly what
    we need for configuration 1.
  \end{proof}

  Continuing on the proof of Theorem \ref{thm:carlb}, we show what
  happens when, starting from configuration 1, \CAR\ processes the
  following request sequence.

  \begin{description}
  \item[{\bf Page $2N+1$:}] A page in $T_2$ is demoted to $B_2$, which
    loses a page; the marked page from $T_1$ is moved to $T_2$ and the
    new page is moved into $T_1$.
  \item[{\bf MRU page in $B_2$:}] This should have decremented $p$ but
    remains unchanged since it is already zero. Since the size of
    $T_1$ is more than $p$ \CAR\ will call {\sc Replace} and $2N+1$
    will be demoted to $B_1$, resulting in configuration $3$ in Table
    \ref{tab:lbcar}.
  \item[{\bf Page $2N+1$:}] It will move to $T_2$, $p$ is increased
    and a page in $T_2$ is demoted to $B_2$. See configuration $4$ in
    Table \ref{tab:lbcar}.
  \item[{\bf MRU page from $B_2$, repeat $N-2$ times:}] It results in
    configuration $N+2$ in Table \ref{tab:lbcar}.
  \end{description}

\begin{table}[bht]
  \centering
  \begin{tabular}[b]{|c|c|c|c|c|c|}
    \hline
    Config. & $p$ & $B_1$ & $T_1$ & $T_2$ & $B_2$ \\
    \hline \hline
    1 & 0  & $\emptyset$ & $1^1$ & $2, \ldots, N$  & $N+1, \ldots, 2N$ \\ \hline  
    2 & 0  & $\emptyset$ & $2N+1$ & $1, \ldots, N-1$  & $N, \ldots, 2N-1$ \\ \hline
    3 & 0  & $2N+1$ & $\emptyset$ & $N,1, \ldots, N-1$  & $N+1, \ldots, 2N-1$ \\ \hline  
    4 & 1  & $\emptyset$ & $\emptyset$ & $2N+1,N,1, \ldots, N-2$  & $N-1,N+1, \ldots, 2N-1$ \\ \hline    5 & 0  & $\emptyset$ & $\emptyset$ & $N-1,2N+1,N,1, \ldots, N-3$  & $N-2, N+1, \ldots, 2N-1$ \\ \hline 
    $\ldots$ & 0  & $\ldots$ & $\ldots$ & $\ldots$  & $\ldots$ \\ \hline 
    $N+2$ & 0  & $\emptyset$ & $\emptyset$ & $2, \ldots, N-1, 2N+1, N$  & $1,N+1, \ldots, 2N-1$ \\ \hline        
  \end{tabular}
  \caption{Example for Lower Bound on \CAR's competitiveness}
  \label{tab:lbcar}
\end{table}

  The request sequence detailed above generates $N+1$ faults for \CAR\ 
  while only $N$ different pages are requested. Thus, \OPT\ could
  limit itself to at most one fault in this stretch. \OPT\ will fault
  once during each stretch if the next page is picked to be one that
  is farthest used in the future. Repeating the above steps an
  unbounded number of times with appropriate relabeling proves that the
  competitiveness ratio of \CAR\ is lower bounded by $N+1$.
\end{proof}

\subsection{Non-Adaptive \ARC\ and \CAR\ are not Competitive}

It is particularly interesting to note that the non-adaptive version
of \CAR\ and \ARC\ (called {\em Fixed Replacement cache})
\cite{mm-ARC-03} are not competitive. The following two theorems prove
that the competitiveness ratios can be unbounded in this case. 

\begin{theorem}
\label{thm:fixp}
Algorithm \CAR\ with fixed $p$ is not competitive.
\end{theorem}
\begin{proof}
  Suppose that algorithm \CAR\ has $p$ fixed instead of being adaptive
  and $0 < p < N-1$. Recall that $p$ is the target size of $T_1$ and
  $N-p$ is the target size of $T_2$. We design a request sequence such
  that with less than $N$ pages we can generate an infinite number of
  page faults for \CAR. The sequence is described as follows:

  \begin{description}
  \item[{\bf Step 1:}] Fill up $T_2$ with $N-p$ unmarked pages as described
    above in the proof of Theorem \ref{thm:carlb}. 
  \item[{\bf Step 2:}] Request the MRU page in $B_2$. The requested page goes
    to the tail of $T_2$ as an unmarked page. Since the size of $T_2$
    is greater than $p$ we discard the head of $T_2$.
  \item[{\bf Step 3:}] Request the MRU page in $B_2$ which is actually
    the page discarded in Step 2 from $T_2$. This step is similar to
    Step 2 and we can continue to repeat this infinitely often, since
    the page that moves from $B_2$ to $T_2$ get's unmarked and the
    page that moves from $T_2$ to $B_2$ goes to MRU.
\end{description}
Therefore, we can cycle infinitely many times through $N-p +1$ pages
triggering an infinite number of faults, while \OPT\ can avoid faults
altogether during the cycle. 
\end{proof}

\begin{theorem}
  \label{thm:fixpARC}
  Algorithm \ARC\ with fixed $p$ is not competitive.
\end{theorem}
	
\begin{proof}

  Suppose that algorithm \ARC\ has $p$ fixed instead of being adaptive
  and $0 < p < N$. Recall that $p$ is the target size of $T_1$ and
  $N-p$ is the target size of $T_2$. We design a request sequence such
  that with less than $N$ pages we can generate an infinite number of
  page faults for \ARC. The first step is to fill up $T_2$ (size of
  $T_2 = N-p$). Next we request the MRU page in $B_2$. Every time we
  request a page from $B_2$, it goes into the top of $T_2$ and thus it
  increases the size of $T_2$ beyond its target size. It follows that
  \ARC \: will call \Replace\ and move a page from $T_2$ to the MRU
  position in $B_2$.  If the MRU page from $B_2$ is repeatedly
  requested, we will cycle through $N-p$ pages, every time incurring a
  page fault for \ARC, while \OPT\ can avoid faults altogether during
  the cycle.
\end{proof}

\section{Analyzing \LRU\ using potential functions}
\label{sec:ub}

\subsection{The generic approach}
The standard approach used here is as follows. First, we define a
carefully crafted potential function, $\Phi$.  As per the strategy of
analyzing competitiveness ratios suggested by Sleator and Tarjan
\cite{st-aelupr-85}, we then try to show the following inequality:
\begin{empheq}[box={\mybluebox[5pt]}]{equation}
  \label{eq:tpt}
  C_A + \Delta\Phi \le f(N) \cdot C_O + g(N),
\end{empheq}
where $C_A$ and $C_O$ are the costs incurred by the algorithm and by
\OPT, respectively, $\Delta\Phi$ is the change in potential, $f(N)$ is
some function of $N$, the size of the cache.

In all of our proofs, we assume that the work involves simultaneously
maintaining \OPT's cache as well as the algorithm's cache. So we can
break down the work into two steps, one where only \OPT\ serves and
one where only the algorithm serves. When only \OPT\ serves, there are
2 cases: first when \OPT\ has a hit and next when \OPT\ has a
miss. Next, we consider the cases when the algorithm serves, once when
it has a hit and once when it has a miss. In each case, our goal is to
prove the inequality (\ref{eq:tpt}) mentioned above, which establishes
that $f(N)$ is the competitiveness ratio of algorithm $A$. There may
be an additive term of $g(N)$ which is a function of the misses needed
to get to some initial configuration for the cache.

\subsection{Analyzing \LRU\ using potential functions}

Assuming that the size of cache given to the competing \OPT\ algorithm
is $N_O \le N$, the following result was proved by Sleator and Tarjan
\cite{st-aelupr-85} (Theorem 6) for \LRU.

\begin{theorem}
  \label{thm:lru}
  \cite{st-aelupr-85} Algorithm \LRU\ is $\big(\frac{N}{N-N_O+1}\big)$-competitive.
\end{theorem}

Here we present a complete proof of this well-known result because we
believe it is instructive for the other proofs in this paper. 

\begin{proof}
  While this was not used in the proof in Sleator and Tarjan
  \cite{st-aelupr-85}, a potential function that will facilitate the
  proof of the above theorem is:
  \begin{empheq}[box={\mybluebox[5pt]}]{equation}
    \label{eq:phiLRU2}
    \Phi = \frac{\sum_{x\in D} r(x)}{N_L-N_O+1},
  \end{empheq}
  where $D$ is the list of items in \LRU's cache but not in \OPT's
  cache, and $r(x)$ is the rank of item $x$ in \LRU's list with the
  understanding that the LRU item has rank 1, while the MRU item has
  rank equal to the size of the cache \cite{a-coa-96}.

  We now show the following inequality:
  \begin{empheq}[box={\mybluebox[5pt]}]{equation}
    \label{eq:tpt-lru}
    C_A + \Delta\Phi \le \big(\frac{N}{N-N_O+1}\big) \cdot C_O + O(N),
  \end{empheq}
  where $C_A$ and $C_O$ are the costs incurred by the algorithm and by
  \OPT, respectively, $\Delta\Phi$ is the change in potential, $f(N)$
  is some function of $N$, the size of the cache.

  We assume that the work involves simultaneously maintaining \OPT's
  cache as well as \LRU's cache. So we can break down the work of
  \LRU\ into two steps, one where only \OPT\ serves and one where only
  \LRU\ serves. When only \OPT\ serves, there are 2 cases: first when
  \OPT\ has a hit and next when \OPT\ has a miss. In either case, the
  cost for \LRU\ is 0, since only \OPT\ is serving. When \OPT\ has a
  hit, the cost for \OPT\ is also 0. Furthermore, since \LRU's cache
  remains untouched, and no changes take place in the contents of
  \OPT's cache, the ranks of the items in \LRU\ remain
  unchanged. Thus, $\Delta\Phi = 0$.  Therefore, the inequality in
  (\ref{eq:tpt-lru}) is trivially satisfied in this case.

  When \OPT\ has a miss, $C_O = 1$, as before. The item evicted by
  \OPT\ can contribute the rank of that item to increase at most by
  $N_L$, making the increase in potential function to be bounded by
  $\frac{N_L}{N_L-N_O+1}$.  Thus, the inequality in
  (\ref{eq:tpt-lru}) is satisfied.  

  Next, we consider the step where \LRU\ serves the request. As with
  \OPT, when \LRU\ is serving, the cost for \OPT\ is 0. We again
  consider two cases: first when \LRU\ has a hit and next when \LRU\
  has a miss. When \LRU\ has a hit, the cost for \LRU\ is 0. The
  contents of \LRU's cache may change. The item that was accessed is
  moved to the MRU position. However, this item is already in \OPT's
  cache and therefore cannot contribute to a change in
  potential. Several other items may move down in the cache, thus
  contributing to a decrease in potential of at most $(N-1)$. In the
  worst, case the increase in potential is at most 0. Therefore, the
  inequality in (\ref{eq:tpt-lru}) is again satisfied.  

  Finally, we consider the case when \LRU\ has a miss. As before,
  $C_L = 1$. Following the previous arguments, an item would be
  brought into MRU (which is already present in \OPT's cache), a bunch
  of items may be demoted in rank, and the \LRU\ item will be
  evicted. The only action that can contribute to an increase is
  caused by the item that is brought into the MRU location. However,
  this item is already present in \OPT's cache, and hence cannot
  contribute to an increase. All the demotions and eviction can only
  decrease the potential function. Note that before the missed item is
  brought into \LRU's cache, the contents of \LRU's and \OPT's cache
  agree in at most $N_O-1$ items, since \OPT\ just finished serving
  the request and the item that caused the miss is already in \OPT's
  cache. Thus there are at least $N_L - N_O +1$ items that contribute
  their ranks to the potential function. These items either get
  demoted in rank or get evicted. Either way, the potential function
  will reduce by a minimum value of $N_L - N_O +1$, although it could
  more if there are more items that are in \LRU\ and that are not in
  \OPT's cache. Thus the total change in potential has to be at most
  $N_L - N_O +1$, and we have
  \begin{eqnarray*}
    C_L + \Delta\Phi \le 1 - \frac{(N_L - N_O +1)}{(N_L - N_O +1)} \le 0 
    = \frac{N_L}{N_L-N_O+1} \cdot C_O.
  \end{eqnarray*}
  
  Summarizing the costs, we have the following:
  \begin{center}
    \begin{tabular}{ | c | c | c | c |}
      \hline
      Step & $C_L$  & $\Delta\Phi$ & $C_O$ \\ 
      \hline\hline
      \multicolumn{4}{|l|}{\bf \OPT\ Serves Request} \\
      \hline
      \OPT\ has a hit & 0 & 0 & 0 \\ \hline
      \OPT\ has a miss & 0 & $\le N_L$ & 1 \\ \hline\hline
      \multicolumn{4}{|l|}{\bf \LRU\ Serves Request} \\ \hline
      \LRU\ has a hit & 0 & $\le 0$ & 0 \\ \hline
      \LRU\ has a miss & 1 & $\le N_L - N_O +1$ & 0 \\ \hline
    \end{tabular}
  \end{center}
\end{proof}

The analysis of \LRU\ states that if the sizes of \LRU's and \OPT's
caches are $N_L$ and $N_O$ respectively, and if $N_L \ge N_O$, then
the competitiveness ratio of \LRU\ is $\frac{N_L}{N_L-N_O+1}$. Thus
\LRU\ is 2-competitive if the size of \LRU's cache is roughly twice
that of \OPT's cache.

\subsection{Analyzing the competitiveness of \CLOCK}

Our result on the competitiveness of \CLOCK\ is formalized in the
following theorem. While this result appears to be known, we have not
been able to locate a full proof and we believe this is of value. We
therefore present it for the sake of completeness.

\begin{theorem}
  \label{thm:clock}
  Algorithm \CLOCK\ is $\big(\frac{N}{N-N_O+1}\big)$-competitive.
\end{theorem}

\begin{proof}
  Let $M_0$ denote the subsequence of unmarked pages in \CLOCK,
  ordered counterclockwise from head to tail. Let $M_1$ denote the
  subsequence of marked pages in \CLOCK, ordered counterclockwise from
  head to tail. Let $q$ be any page in \CLOCK's cache. Let $P^0[q]$
  denote the position of an unmarked page $q$ in the ordered sequence
  $M_0$, and let $P^1[q]$ denote the position of a marked page $q$ in
  $M_1$. Finally, let $R[q]$ denote the rank of page $q$ defined as
  follows:
  \begin{empheq}[box={\mybluebox[5pt]}]{equation}
    R[q] = 
    \begin{cases} 
      P^0[q] & \text{if $q$ is unmarked,} \\
      P^1[q] + |M_0| &\text{otherwise.} 
    \end{cases}
  \end{empheq}
  Thus, if $q$ is an unmarked page at the head, then $R[q] = 1$. By
  the above definition, the following lemmas are obvious.

  \begin{lemma}
    \label{lem:CL-rankBound}
    If $q$ is any page in \CLOCK's cache, then $R[q] \leq N$.
  \end{lemma}
    
  \begin{lemma}
    \label{lem:CL-rankChange}
    If a marked page $q$ at the head of \CLOCK's cache is unmarked and
    moved to the tail, then $R[q]$ does not change in the process.
  \end{lemma}
    
  Let $D$ be the set of pages that are in the cache maintained by
  \CLOCK, but not in the cache maintained by \OPT. We define the
  potential function as follows:
  \begin{empheq}[box={\mybluebox[5pt]}]{equation}
    \label{eq:clockPhi} \Phi = \sum_{q\in D} R[q]
  \end{empheq}

  We prove one more useful lemma about the ranks as defined above. 

  \begin{lemma}
    \label{lem:CL-phiChange}
    If an unmarked page at the head of \CLOCK's cache is evicted from
    \CLOCK's cache, and if there is at least one page in $D$, then
    $\Phi$ decreases by at least 1 in the process. .
  \end{lemma}
    
  \begin{proof}
    All pages, marked or unmarked, will move down by at least one
    position (reducing the rank of each by at least 1). The decrease
    in potential for at least one page that is in $D$ will contribute
    to $\Phi$, guaranteeing that $\Delta\Phi \leq -1$.
  \end{proof}

  Let $C_{\CLOCK}$ and $C_{\OPT}$ be the costs incurred by the
  algorithms \CLOCK\ and \OPT, and let
  ${\cal S} = \sigma_1, \sigma_2, \ldots, \sigma_m$ be an arbitrary
  request sequence. Let ${\cal S}'$ denote the initial subsequence of
  requests that take place prior to the cache becoming full. Note that
  exactly $N$ faults are incurred in ${\cal S}'$, after which the
  cache remains full. Let ${\cal S}''$ be the subsequence of ${\cal S}$
  that comes after ${\cal S}'$.

  Let $C_{\CLOCK}$ and $C_{\OPT}$ be the cost incurred by the
  algorithms \CLOCK\ and \OPT\ respectively.  We will prove that for
  every individual request, $\sigma \in {\cal S}''$:
  \begin{equation} \label{eq:3cl}
    C_{\CLOCK}(\sigma) + \Delta \Phi \leq N *C_{\OPT}(\sigma)
  \end{equation}

  As before, we assume that request $\sigma$ is processed in two
  distinct steps: first when \OPT\ services the page request and, next
  when \CLOCK\ services the request. We will show that inequality
  (\ref{eq:3cl}) is satisfied for both the steps.

  When only \OPT\ acts in this step, $C_{clock} = 0$. If \OPT\ does
  not fault on this request, then $C_{OPT} = 0$. No change occurs to
  the contents of the cache maintained by \OPT\ as well as \CLOCK, and
  the clock hand does not move. Thus, $\Delta \Phi = 0$, satisfying
  inequality \ref{eq:3cl}. 

  If \OPT\ faults on request $\sigma$, then $C_{OPT} = 1$ and
  $C_{\CLOCK} = 0$. The contents of the cache maintained by \OPT\ does
  change, which could affect the potential function. The potential
  could increase due to the eviction of a page in \OPT. Since by Lemma
  \ref{lem:CL-rankBound} the rank of the evicted page cannot exceed $N$,
  the potential will change by at most $N$. Thus, inequality
  \ref{eq:3cl} is satisfied.

  Next we consider what happens when \CLOCK\ services the request. For
  this case $C_{OPT} = 0$. If \CLOCK\ does not fault, then
  $C_{clock} = 0$ and the requested page may change from an unmarked
  status to a marked one. However, since the page is already in the
  cache maintained by \OPT\ it is not in $D$ and is therefore not
  considered for the potential function calculations in
  \ref{eq:clockPhi}. Thus, inequality \ref{eq:3cl} is satisfied.

  Finally, we consider the case when \CLOCK\ faults, in which case
  $C_{\CLOCK} = 1$ and $C_{\OPT} = 0$. To satisfy inequality
  \ref{eq:3cl}, $\Delta\Phi$ needs to be less or equal to -1.  When
  \CLOCK\ has a miss, if the head page happens to be marked, then
  \CLOCK\ will repeatedly unmark the marked head page, moving it to
  the tail position, until an unmarked head page is encountered. The
  unmarked head page is then evicted. Each time a marked head page
  becomes an unmarked tail page, by Lemma \ref{lem:CL-rankChange} its
  rank does not change. When finally an unmarked head page is evicted,
  we know that there is at least one page in \OPT's cache that is not
  in \CLOCK's cache (i.e., the page that caused the fault). Since
  there are N pages in the cache maintained by \CLOCK, at least one of
  those pages is guaranteed not to be part of the cache maintained by
  \OPT. Since there is at least one page in $D$, by Lemma
  \ref{lem:CL-phiChange} it is clear that evicting an unmarked head page
  will decrease the potential function by at least one, which will
  pay for the \CLOCK's page fault.

  We have therefore showed that for every request $\sigma$, inequality
  \ref{eq:3cl} is satisfied. Since there can be at most $N$ faults for
  the requests in ${\cal S'}$, summing up the above inequality for all
  requests, $\sigma \in {\cal S}$, we get 
  \[ C_{\CLOCK}({\cal S}) \leq N *C_{\OPT}({\cal S}) +
    N.\] 
  This completes the proof of the theorem and the
  competitiveness analysis of the \CLOCK\ algorithm.
\end{proof} 


\subsection{Analyzing the Competitiveness of ARC}
\label{sec:anal}

In this paper, we prove two different upper bounds for the
competitiveness of \ARC. These two proofs use very different potential
function. The first one allows for the sizes of the caches maintained
by \ARC\ and \OPT\ to be different, while the second one does not
allow for it, but provides a tighter bound. We provide both results
below. 

Our first result on the competitiveness of \ARC\ is formalized in the
following theorem:

\begin{theorem}
  \label{thm:arc}
  Algorithm \ARC\ is $\big(\frac{12N}{N-N_O+1}\big)$-competitive.
\end{theorem}

\begin{proof}
  Let $P_X[q]$ be the position of page $q$ in an arbitrary ordered
  sequence of pages $X$. When the set is obvious, we will drop the
  subscript and denote $P_X[q]$ simply by $P[q]$. The set of history
  pages $T_1$, $T_2$, $B_1$, and $B_2$ will be treated as an ordered
  sequence of pages ordered from its LRU position to its MRU
  position. Let $\OPT$ and $\CAR$ be the set of main pages stored in
  the caches for algorithms \OPT\ and \ARC\ respectively.
  Let 
  $D = \ARC \setminus \OPT$.
  %
  %
  As before, we associate each page with a rank value $R[q]$, which is
  defined as follows:
  \begin{empheq}[box={\mybluebox[5pt]}]{equation}
    R[q] = 
    \begin{cases} 
      2P_{B_1}[q]  & \text{if } q \in B_1  \\
      2P_{B_2}[q]  & \text{if } q \in B_2  \\
      4 P_{T_1}[q] + 2b_1  &\text{if } q \in T_1 \\
      4 P_{T_2}[q] + 2b_2  &\text{if } q \in T_2
    \end{cases}
    \label{eq:R2}
  \end{empheq}
  Finally, we define the potential function as follows:
  \begin{empheq}[box={\mybluebox[5pt]}]{equation}
    \Phi = p+2t_1+ 2\big(\frac{\sum_{q \in D} R[q]}{N - N_O + 1}\big) - 3
    |\ARC| 
    \label{eq:Phi2}
  \end{empheq}
The initial value of $\Phi$ is 0. If the following inequality
(\ref{eq:3}) is true for any request $\sigma$, where $\Delta \Phi$
is the change in potential caused by serving the request, then when
summed over all requests, it proves Theorem \ref{thm:arc}.
\begin{equation}
  C_{\ARC}(\sigma) + \Delta \Phi \leq \frac{12 N C_{\OPT}(\sigma)}{N - N_O + 1}. 
  \label{eq:3}
\end{equation}
As before, we assume that request $\sigma$ is processed in two
distinct steps: first when \OPT \: serves and, next when \ARC \:
serves. We will show that inequality (\ref{eq:3}) is satisfied for
each of the two steps.

\subsubsection*{Step 1: \OPT \: serves request $\sigma$}

Since only \OPT \: acts in this step, $C_{\ARC} = 0$, and
$T_1 \cup T_2$ does not change. There are two possible cases: either
\OPT \: faults on $\sigma$ or it does not. If \OPT \: does not fault
on this request, then it is easy to see that $C_{\OPT} = 0$ and
$\Delta \Phi = 0$, thus satisfying inequality (\ref{eq:3}).

If \OPT \: faults on request $\sigma$, then $C_{\OPT} = 1$ and some
page $q$, is evicted from the cache maintained by \OPT\ will belong to
$D$ after this step and thus its rank will contribute to the potential
function, which will increase by two times the rank of $q$. The
maximal positive change in potential will occur when $q$ is the MRU
page of either $T_1$ or $T_2$. In this case the rank of $q$ is given
by: $R[q] =4P[q] + b_1$ ($R[q] =4P[q] + b_2$). The maximum possible
values for each of the terms $P[q]$ and $b_1$ will be $N$, hence the
maximum possible rank of $q$ will be $4N + 2N= 6N$. Therefore
resulting potential change is at most 
$\frac{12N}{N-N_O+1}$.

\subsubsection*{Step 2: \ARC \: serves request $\sigma$}

We break down the analysis into four cases. Case 2.1 deals with the
case when \ARC \: finds the page in its cache. The other three cases
assume that \ARC \: faults on this request because the item is not in
$T_1 \cup T_2$. Cases 2.2 and 2.3 assume that the missing page is
found recorded in the history in lists $B_1$ and $B_2$,
respectively. Case 2.4 assumes that the missing page is not recorded
in history.

\paragraph*{Case 2.1: \ARC \: has a page hit}

Clearly, the page was found in $T_1 \cup T_2$, and $C_{\ARC} = 0$. We
consider the change of each of terms in the potential function
individually.

\begin{enumerate}
\item As per the algorithm, $p$ can only change when the page is found
  in history. (See lines 3 through 10 of \ARC$(x)$.) Since the page is
  not found in \ARC's history, $\Delta p = 0$.
\item If the hit happens in $T_1$, the page will move to the top of
  $T_2$ (See line 2 of \ARC$(x)$.), which will result in a decrease in
  $t_1$. If the hit happens in $T_2$, the size of $t_1$ will remain
  the same. The overall change in $t_1$ will be 0.
\item Since \OPT \: has already served the page, the page is in \OPT's
  cache. Therefore, even if the page's rank could change when moved
  from $T_1$ to MRU position of $T_2$, this rank will not affect the
  potential since the page is not in $D$.
\end{enumerate}
We, therefore, conclude that $ \Delta \Phi = 0$, satisfying inequality
(\ref{eq:3}).

Next we will analyze the 3 cases when the requested page is not in
\ARC's cache. Since $C_{\ARC}= 1$, the change in potential must be
$\le -1$ in each case in order for inequality (\ref{eq:3}) to be
satisfied.

\paragraph*{Case 2.2: \ARC \: has a page miss and the missing page
  is in $B_1$}
We consider the two cases, first when \Replace\ moves an item from
$T_1$ to $B_1$ and second when it moves an item from $T_2$ to $B_2$.
\begin{enumerate} \topsep=0pt\itemsep=0pt
\item Case 1: We consider the change in potential function by
  analyzing each of the 3 terms.
\begin{itemize}
\item Value of $p$ will either increase by $1$ or stay the same in
  case $p=N$, we will account for the worst case which is when
  $\Delta p = 1$.
\item A new page is being added to MRU of $T_2$, and \Replace\ is
  taking the LRU page of $T_1$ to $B_1$, then $2\Delta t_1 = -2$.
\item The page that moved from $B_1$ to $T_2$ is not in $D$, therefore
  the change in its rank will not affect the potential, the other
  pages will could only decrease their rank, meaning that
  $2\Delta \sum_{q \in D} R[q] \leq 0$.
\end{itemize}
Since $p$ increases by at most 1 and $t_1$ decreases by at least 2 the
total change in potential is at most -1.
\item Case 2: Once again. we consider the change in potential function
  by analyzing each of the three terms.
\begin{itemize}
\item Value of $p$ will either increase by $1$ or stay the same in
  case $p=N$, we will account for the worst case which is when
  $\Delta p = 1$.
\item A new page is added to MRU of $T_2$, and \Replace\ moves the LRU
  page of $T_2$ to $B_2$. Thus, there is no change in $T_1$.
\item The page that moved from $B_1$ to $T_2$ is not in $D$, therefore
  the change in its rank will not affect the potential. Since
  $t_1 + t_2 = N$, it is guaranteed that at least $N-N_O + 1$ pages
  are not in \OPT. For the pages that are in $T_1$, their ranks will
  decrease by at least 2 since $b_1$ decreases by 1, and for the pages
  in $T_2$ their ranks will decrease by at least 2 as well since $b_2$
  increases by 1 but the LRU page in $T_2$ will move to $B_2$,
  reducing $P[q]$ for all the pages in $T_2$. The term
  $2\frac{\sum_{q \in D} R[q]}{N - N_O + 1}$ decreases by at least -4.
\end{itemize}
Since $p$ increases by at most 1 and
$2\frac{\sum_{q \in D} R[q]}{N - N_O + 1}$ decreases by at least -4
the total change in potential is at most -3.
\end{enumerate}

\paragraph*{Case 2.3: \ARC \: has a page miss and the missing page
  is in $B_2$}

When the missing page is in $B_2$, \ARC\ makes a call to {\sc Replace}
(Line 5) and then executes Lines 18-19. Thus, $p$ is decremented
except if it is already equal to 0.
We consider two sub cases: $\Delta p \leq -1$ and $\Delta p = 0$.
\paragraph*{\underline{$\Delta p \leq -1$}: }
As in Case 2.2, the call to {\sc Replace} has no effect on
$t_1$. \Replace\ will not increment the rank using a similar analysis
as in 2.2 and change in $p$ will at least be -1.  The change in the
potential function is at most -1.

\paragraph*{\underline{$\Delta p = 0$}: }
Unlike the sub case above when $p$ decreases by 1, the change in $p$
cannot guarantee the required reduction in the potential. We therefore
need a tighter argument. 
We know that there is a call to {\sc Replace}. Two cases arise and are
discussed below.
\begin{itemize} \topsep=0pt\itemsep=0pt
\item \Replace\ moves an item in $T_1$ to $B_1$: Since the LRU page of $T_1$ is
  moved to the MRU position of $B_1$, $2\Delta t_1 = -2$ and there is
  no movement of a page in $D$ that could increase the rank. Therefore
  the total change in the potential function is at most -2.

\item \Replace\ moves an item in $T_2$ to $B_2$: $p=0$ indicates that
  $T_2$ has $N$ pages, therefore is guarantee that at least $N-N_O+1$
  pages will not be part of \OPT, contributing to the change in
  potential. The page being moved from $T_2$ to $B_2$ will decrease
  it's rank by at least 2, and the rest of the pages in $T_2$ will
  move down one position ($P[q]$ will decrease by 1) while $B_2$ will
  remain the same, resulting in a change in the potential function of
  at most -4.
\end{itemize}
Thus, in each case the potential function decreased by at most -2.

\paragraph*{Case 2.4: \ARC \: has a page miss and the missing page
  is not in $B_1 \cup B_2$}

\begin{enumerate}

\item $t_1+b_1=N$; $t_1<N$; The LRU page in $B_1$ is evicted. Assume
  \Replace\ moves a page from $T_1$ to $B_1$ and a new page is brought
  into $T_1$ ($\Delta t_1 = 0$, $\Delta b_1 = 0$, $\Delta t_2 = 0$,
  $\Delta b_2 = 0$).

  \begin{itemize}
  \item The term $p$ is not affected.
  \item The term $t_1$ is not affected.
  \item Since $t_1+b_1 = N$, at least $N-N_o+1$ pages in
    $T_1 \cup B_1$ are not in $\OPT$. If the page is in
    $B_1 \setminus \OPT$ then its rank decreases by $2$; if the page
    is in $T_1 \setminus \OPT$ its rank decreases by $4$.
  \end{itemize}
    
\item $t_1+b_1=N$; $t_1<N$; The LRU page in $B_1$ is evicted. Assume
  \Replace\ moves a page from $T_2$ to $B_2$ and a new page is brought
  into $T_1$ ($\Delta t_1 = 1$, $\Delta b_1 = -1$, $\Delta t_2 = 1$,
  $\Delta b_2 = 1$).

  \begin{itemize}
  \item The term $p$ is not affected.
  \item The term $t_1$ is increased by $1$.
  \item Since $t_1+t_2 = N$, at least $N-N_o+1$ pages in
    $T_1 \cup T_2$ that are not in \OPT. If a page, $q$, is in
    $T_1 \setminus \OPT$ then its rank decreases by $2$
    ($\Delta R[q] = \Delta 4*P[q] + \Delta 2*b_2 = -2$); if the page,
    $q$, is in $T_2 \setminus \OPT$ its rank decreases by $2$
    ($\Delta R[q] = \Delta 4*P[q] + \Delta 2*b_2 = -2$).
  \end{itemize}

\item $t_1+b_1<N$; $t_1+t_2+b_1+b_2=2N$; Assume that the LRU page in
  $B_2$ is evicted and \Replace\ moves a page from $T_1$ to $B_1$ and
  a new page is brought into $T_1$ ($\Delta t_1 = 0$,
  $\Delta b_1 = 1$, $\Delta t_2 = 0$, $\Delta b_2 = -1$).

  \begin{itemize}
  \item The term $p$ is not affected.
  \item The term $t_1$ is not affected.
  \item Here we used the fact that $t_2+b_2 > N$, then at least
    $N-N_o+1$ pages in $T_2 \cup B_2$ are not in $\OPT$. If a page, $q$,
    is in $T_2 \setminus \OPT$ then its rank decreases by $2$
    ($\Delta R[q] = \Delta 4*P[q] + \Delta 2*b_2 = 4 *(0) + 2(-1) =
    -2$); if the page, $q$, is in $B_2 \setminus \OPT$ its rank
    decreases by $2$ ($\Delta R[q] = \Delta 2*P[q] = 2 * (-1)= -2$).
  \end{itemize}

\item $t_1 + b_1<N$; $t_1 + t_2+b_1 + b_2 = 2N$; Assume that the LRU
  page in $B_2$ is evicted and \Replace\ moves a page from $T_2$ to
  $B_2$ and a new page is brought into $T_1$ ($\Delta t_1 = 1$,
  $\Delta b_1 = 0$, $\Delta t_2 = 1$, $\Delta b_2 = 0$).

  \begin{itemize}
  \item The term $p$ is not affected.
  \item The term $t_1$ is increased by $1$.
  \item Here we used the fact that $t_2+b_2 > N$, then at least
    $N-N_o+1$ pages in $T_2 \cup B_2$ are not in $\OPT$. If a page, $q$,
    is in $T_2 \setminus \OPT$ then its rank decreases by $2$
    ($\Delta R[q] = \Delta 4*P[q] + \Delta 2*b_2 = 4 *(0) + 2(-1) =
    -2$); if the page, $q$, is in $B_2 \setminus \OPT$ its rank
    decreases by $2$ ($\Delta R[q] = \Delta 2*P[q] = 2 * (-1)= -2$).
  \end{itemize}


\item $t_1 + b_1<N$; $t_1 + t_2+b_1 + b_2 < 2N$; In this case, no
  pages are evicted from history. Assume that \Replace\ moves a page
  from $T_1$ to $B_1$ and a new page is brought into $T_1$
  ($\Delta t_1 = 0$, $\Delta b_1 = 1$, $\Delta t_2 = 0$,
  $\Delta b_2 = 0$)

  \begin{itemize}
  \item The term $p$ is not affected.
  \item The term $t_1$ is increased by $1$.
  \item Here we cannot say that the rank decreases. Hence the rank
    term is at most $0$.
  \item The term $|\ARC|$ increases by $1$.
  \end{itemize}

\item $t_1 + b_1<N$; $t_1 + t_2+b_1 + b_2 < 2N$; In this case, no
  pages are evicted from history. Assume \Replace\ moves a page from
  $T_2$ to $B_2$ and a new page is brought into $T_1$
  ($\Delta t_1 = 1$, $\Delta b_1 = 0$, $\Delta t_2 = -1$,
  $\Delta b_2 = 1$)

  \begin{itemize}
  \item The term $p$ is not affected.
  \item The term $t_1$ is not affected.
  \item Here we cannot say that the rank decreases. Hence the rank term is at most $0$.
  \item The term $|\ARC|$ increases by $1$.
  \end{itemize}
  
\end{enumerate}


\paragraph*{Wrapping up the proof of Theorem \ref{thm:arc}: }
Combining the four cases (2.1 through 2.4) proves that inequality
(\ref{eq:3}) is satisfied when \ARC\ serves request $\sigma$. This
completes the proof of Theorem \ref{thm:arc}, establishing that the
upper bound on the competitiveness of \ARC\ is $12N$ for the cases
where the sizes of \OPT\ and \ARC\ are the same. By analyzing cases
where the size of \ARC\ is greater than \OPT\ we can observe that
since \ARC\ will be $\frac{12N}{N - N_O + 1}$ the greater the size of
\ARC's cache relative to the size of \OPT's cache, smaller will be the
competitiveness of \ARC.
\end{proof} 

\subsection{Alternative Analysis of Competitiveness of \ARC}
\label{sec:altarc}

Below, we prove an improved upper bound on the competitiveness ratio
of \ARC. As seen below, the potential function is considerably
different. Let $C_A$ and $C_O$ be the costs incurred by the algorithms
\ARC\ and \OPT. 

We start with some notation and definitions. If $X$ is the set of
pages in a cache, then let $MRU(X)$ and $LRU(X)$ be the most recently
and least recently used pages from $X$. Let $MRU_k(X)$ and $LRU_k(X)$
be the $k$ most recently and $k$ least recently used pages from $X$.

Let lists $L_1$ (and $L_2$) be the lists obtained by concatenating
lists $T_1$ and $B_1$ ($T_2$ and $B_2$, resp.). Let list $L$ be
obtained by concatenating lists $L_1$ and $L_2$. We let
$\ell_1, \ell_2, t_1, t_2, b_1, b_2$ denote the sizes of
$L_1, L_2, T_1, T_2, B_1, B_2$, respectively. Finally, let
$t := t_1 + t_2$ and $\ell := \ell_1 + \ell_2$.

At any instant of time during the parallel simulation of \OPT\ and
\ARC, and for any list $X$, we let $MRU_k(X)$ be denoted by $TOP(X)$,
where $k$ is the largest integer such that all pages of $MRU_k(X)$ are
also in the cache maintained by OPT. We let $L_1', L_2', T_1', T_2'$
denote the $TOP$s of $L_1, L_2, T_1, T_2$, respectively, with sizes
$\ell_1', \ell_2', t_1', t_2'$, respectively. We let $b_1'$ and $b_2'$
denote the sizes of the $B_1' = L_1' \cap B_1$ and
$B_2' = L_2' \cap B_2$, respectively. Note that if $b_1'> 0$
($b_2' > 0$, resp.), then all of $T_1$ ($T_2$, resp.) is in
\OPT. Finally, we let $\ell' := \ell_1' + \ell_2'$.
The \ARC\ algorithm ensures that $0 \le t \le N$, $0 \le \ell \le 2N$
and $0 \le \ell_1 \le N$, thus making $0 \le \ell_2 \le 2N$.

We assume that algorithm $X$ being analyzed is provided an arbitrary
request sequence $\sigma = \sigma_1, \sigma_2, \ldots,
\sigma_m$. 
We define the potential function as follows:
\begin{empheq}[box={\mybluebox[5pt]}]{equation}
  \label{eq:1}
    \Phi = p - (b_1' + 2 \cdot t_1' + 3 \cdot b_2' + 4 \cdot t_2').
\end{empheq}

\noindent
The main result of this section is the following theorem:
\begin{theorem}
  \label{thm:1}
  Algorithm ARC is $4N$-competitive. 
\end{theorem}
We say that the cache is full if $t = N$ and either $t_1 + b_1 = N$ or
$t_2 + b_2 \ge N$. We will prove the above theorem by proving the
following inequality for any request $\sigma$ that is requested after
the cache is full:
\begin{eqnarray} 
  \label{eq:2arc}
  C_A ( \sigma ) + \Delta \Phi \leq 4N \cdot C_O (
  \sigma ) + 2N,
\end{eqnarray}
\noindent where $\Delta X$ represents the change in any quantity
$X$. Summing up the above inequality for all requests 
would prove the theorem as long as the number of faults 
prior to the cache becoming full is bounded by the additive term
$2N$. 

We make the following useful observation about a full cache.

\begin{lemma}
  \label{lem:0}
  When the request sequence requests the $N$-th distinct page, we have
  $t=N$, and this remains an invariant from that point onward.  No
  items are discarded from the cache (main or history) until either
  $t_1 + b_1 = N$ or $\ell_1 + \ell_2 = 2N$. By the time the request
  sequence requests the $2N$-th distinct page, we have either
  $t_1 + b_1 = N$ or $\ell_1 + \ell_2 = 2N$.
\end{lemma}

\begin{proof} 
  Once the request sequence requests the $N$-th distinct page, it is
  obvious that we will have $t=N$, since until then, no item is
  evicted from $T_1 \cup T_2 \cup B_1 \cup B_2$. (Note that {\sc
    Replace} only moves items from the main part to the history, i.e.,
  from $T_1 \cup T_2 $ to $B_1 \cup B_2$.) Also, until then, $p$ does
  not change. From that point forward, the algorithm never evicts any
  item from $T_1 \cup T_2$ without replacing it with some other
  item. Thus, $t=N$ is an invariant once it is satisfied. The history
  remains empty until the main cache is filled, i.e., $t=N$.

  From the pseudocode it is clear that items are discarded from the
  cache in statements 14, 17, and 21; no discards happen from the
  cache until either $t_1 + b_1 = N$ (statement 12) or
  $\ell_1 + \ell_2 = 2N$ (statement 20). If $\ell_1 + \ell_2 = 2N$ is
  reached, since $t_1 + b_1 \le N$, we are guaranteed that
  $t_2 + b_2 \ge N$ and $b_1 + b_2 = N$, both of which will remain
  true from that point onward. Thus, by the time the $2N$-th distinct
  page is requested, we have reached either $t_1 + b_1 = N$ or
  $\ell_1 + \ell_2 = 2N$.
\end{proof}

We assume that request $\sigma$ is processed in two distinct steps:
first when \OPT\ services the page request and, next when \ARC\
services the request. We will show that inequality (\ref{eq:2arc}) is
satisfied for each of the two steps.

\subsubsection*{Step 1: \OPT\ {\em services request $\sigma$}}

Since only \OPT\ acts in this step, $C_A = 0$, and the contents of
\ARC's cache does not change. There are two possible cases: either
\OPT\ faults on $\sigma$ or it does not. Assume that page $x$ is
requested on request $\sigma$.

If \OPT\ does not fault on this request, then $C_O = 0$.  Since the
contents of the cache maintained by \OPT\ does not change, and neither
do the lists $L_1$ and $L_2$, we have $\Delta \Phi= 0$, and
$C_A ( \sigma) + \Delta \Phi \leq 4N \cdot C_O ( \sigma ) \leq 0$.

If \OPT\ faults on request $\sigma$, then $C_O = 1$. The contents
of the cache maintained by \OPT\ does change, which will affect the
potential function. 
\OPT\ will bring in page $x$ into its cache. Assume that it evicts
page $y$ from its cache. The entry of page $x$ into \OPT's cache can
only decrease the potential function. The exit of page $y$ from \OPT's
cache can increase the potential function by at most $4N$. The reason
is as follows. 
Since the sum of $b_1',b_2',t_1',t_2'$ cannot exceed the size of
\OPT's cache, we have $0 \le b_1'+ t_1'+ b_2'+ t_2'\le N$.  Since
$b_1'+ 2t_1'+ 3b_2'+ 4t_2' \le 4(b_1'+ t_1'+ b_2'+ t_2')$, the left
hand side cannot decrease by more than $4N$.  
Thus,
$C_A ( \sigma ) + \Delta \Phi_1 \leq 4N$, proving inequality
(\ref{eq:2arc}).

\subsubsection*{Step 2: \ARC\ {\em services request $\sigma$}}

There are four possible cases, which correspond to the four cases in
\ARC's replacement algorithm. Case 1 deals with the case when \ARC\
finds the page in its cache. The other three cases assume that \ARC\
faults on this request because the item is not in $T_1 \cup T_2$. Cases
2 and 3 assume that the missing page is found recorded in the history
in lists $B_1$ and $B_2$, respectively. Case 4 assumes that the
missing page is not recorded in history.

\smallskip\noindent
{\bf Case I}: \ARC\ has a page hit.

\noindent Clearly, $C_A = 0$. 
We consider several subcases. In each case, the requested page will be
moved to $MRU(T_2)$ while shifting other pages in $T_2$ down.
\begin{description}
\item[Case I.1] If the requested page is in $T_1'$, the move of this page from
  $T_1'$ to $T_2'$ implies $\Delta t_1' = -1; \Delta t_2' = +1$ and
  $\Delta \Phi = -(2 \cdot \Delta t_1' + 4 \cdot\Delta t_2') = -2$.
\item[Case I.2] If the requested page is in $T_2'$, the move of this page to
  $MRU(T_2)$ does not change the set of items in $T_2'$. Thus, $\Delta
  t_1' = \Delta t_2' = 0$ and $\Delta \Phi = 0$.
\item[Case I.3] If the requested page is in $T_1 - T_1'$, then
  $\Delta t_1' = 0; \Delta t_2' = +1$ and $\Delta \Phi = -4$. One
  subtle point to note is that moving $x$ from $T_1-T_1'$ could
  potentially increase $t_1'$ if the following conditions are met: $x$
  is located just below $T_1'$ in $T_1$, it is not in \OPT's cache,
  and the items in $T_1$ immediately below it are in \OPT. However,
  $x$ is already in \OPT's cache and there must be some item above it
  in $T_1$ that is not in \OPT.
\item[Case I.4] If the requested page is in $T_2 - T_2'$, then
  $\Delta t_2' = +1$ and $\Delta \Phi = -4$. The subtle point
  mentioned in Case I.3 also applies here. 
\end{description}

Next we will analyze the three cases when the requested page is not in
\ARC's cache. Since $C_A = 1$, the change in potential must be at
most -1 in order for inequality (\ref{eq:2arc}) to be satisfied. We make
the following useful observations in the form of lemmas.

\begin{lemma}
  \label{lem:a} 
  If \ARC\ has a miss and if the page is not in \ARC's history, we
  have $\ell' = t_1' + t_2' + b_1' + b_2' < N$. Consequently, we also
  have $\ell_1'<N$ and $\ell_2'<N$.
\end{lemma}

\begin{proof} 
  Since \OPT\ has just finished serving the request, the page is
  present in the cache maintained by \OPT\ just before \ARC\ starts to
  service the request. If \ARC\ has a miss, there is at least one page
  in the cache maintained by \OPT\ that is not present in the cache
  maintained by \ARC, implying that $l' < N$. By definition,
  $\ell' = \ell_1' + \ell_2' = t_1' + t_2' + b_1' + b_2'$. Thus, the
  lemma holds.
\end{proof}

\begin{lemma}
  \label{lem:b}
  A call to procedure {\sc Replace} either causes an element to be
  moved from $T_1$ to $B_1$ or from $T_2$ to $B_2$. In either case,
  the change in potential due to {\sc Replace}, denoted by
  $\Delta\Phi_R$, has an upper bound of 1.
\end{lemma}

\begin{proof}
  Procedure {\sc Replace} is only called when \ARC\ has a page
  miss. Clearly, it causes an item to be moved from $T_1$ to $B_1$ or
  from $T_2$ to $B_2$. If that item is in $T_1'$ (or $T_2'$), then
  $T_1 = T_1'$ ($T_2 = T_2'$, resp.) and the moved item becomes part
  of $B_1'$ ($B_2'$, resp.). Because the coefficients of $b_1'$ and
  $t_1'$ ($b_2'$ and $t_2'$, resp.) differ by 1, we have
  $\Delta\Phi_R = +1$. On the other hand, if that element is in
  $T_1 - T_1'$ ($T_2 - T_2'$, resp.), then $B_1'$ ($B_2'$, resp.) was
  empty before the move and remains empty after the move, and thus,
  $\Delta\Phi_R = 0$.
\end{proof}

\begin{lemma}
  \label{lem:d}
  On an \ARC\ miss after phase $P(0)$, if $T_1 = T_1'$ then the {\sc
    Replace} step will not move a page from $T_2'$ to $B_2$. On the
  other hand, if $T_2 = T_2'$ then {\sc Replace} will not move a page
  from $T_1'$ to $B_1$.
\end{lemma}

\begin{proof}
  In an attempt to prove by contradiction, let us assume that
  $T_1 = T_1'$ and $T_2 = T_2'$ are simultaneously true and \ARC\ has
  a miss. By Lemma \ref{lem:0}, we know that after phase, we have
  $t = t_1 + t_2 = N$, which by our assumption means that
  $t_1' + t_2' = N$; this is impossible by Lemma \ref{lem:a}.
  Thus, if $T_1 = T_1'$, then $T_2 \ne T_2'$.  Consequently, if
  $LRU(T_2)$ is moved to $B_2$, this item cannot be from $T_2'$. By a
  symmetric argument, if $T_2' = T_2$, then $T_1 \ne T_1'$, and
  $LRU(T_1)$ is not in $T_1'$.
\end{proof}

\smallskip\noindent
{\bf Case II}: \ARC\ has a miss and the missing page is in $B_1$

\noindent 
Note that in this case the value of $p$ will change by $+1$, unless its value
equals $N$, in which case it has no change. Thus $\Delta p \le 1$. 

If the missing item is in $B_1'$, then $\Delta b_1' = -1$ and $\Delta
t_2' = +1$. Adding the change due to {\sc Replace}, we get 
\begin{eqnarray*}
  \Delta \Phi & \le & 1 - (\Delta b_1' + 4 \cdot \Delta t_2') + \Delta\Phi_R\\
  & \le & -1
\end{eqnarray*}
If the missing item is in $B_1 - B_1'$, then we have $\Delta t_2' = 1$
and $\Delta b_1' = 0$. Thus, we have
\begin{eqnarray*}
  \Delta \Phi & \le & 1 - (\Delta b_1' + 4 \cdot \Delta t_2') + \Delta\Phi_R\\
  & \le & -2
\end{eqnarray*}

\smallskip\noindent
{\bf Case III}: \ARC\ has a miss and the missing page is in $B_2$.

\noindent
Note that in this case the value of $p$ will change by -1, if its value
was positive, otherwise it has no change. Thus $\Delta p \le 0$. 

If the requested item is in $B_2'$, then $\Delta t_2' = 1$, and
$\Delta b_2' = -1$. Thus, we have
\begin{eqnarray*}
  \Delta \Phi & = & \Delta p - (3 \cdot \Delta b_2' + 4 \cdot \Delta t_2') +
  \Delta \Phi_R \\
  & \le & 0
\end{eqnarray*}

But this is not good enough since we need the potential change to be
at most -1. When $\Delta p = -1$, then we get the required inequality
$\Delta\Phi \le -1$. Clearly, the difficulty is when $\Delta p = 0$,
which happens when $p=0$.  Since the missing item is from $b_2'$, it
implies that $B_2'$ is non-empty and $T_2' = T_2$. By Lemma
\ref{lem:d} above, there must be at least one item in $T_1 - T_1'$,
which means that means that $t_1 > 0$. As per the algorithm, since
$T_1$ is non-empty and $p = 0$, we are guaranteed to replace
$LRU(T_1)$, and not an element from $T_1'$. Therefore, {\sc Replace}
will leave $t_1'$ and $b_1'$ unchanged, implying that
$\Delta \Phi_R = 0$. Thus, we have
\begin{eqnarray*}
  \Delta \Phi & = & \Delta p - (3 \cdot \Delta b_2' + 4 \cdot \Delta t_2') +
  \Delta \Phi_R \\
  & \le & -1
\end{eqnarray*}
 
\noindent If the requested item is from $B_2 - B_2'$, then 
$\Delta t_2' = 1$, and $\Delta b_2' = 0$. Thus, we have
\begin{eqnarray*}
  \Delta \Phi & \leq & \Delta p - (4 \cdot \Delta t_2')
  + \Delta \Phi_R \\
  & \leq & -3
\end{eqnarray*}

\smallskip\noindent
{\bf Case IV}: \ARC\ has a miss and the missing page is not in $B_1
\cup B_2$

\noindent We consider two cases. First, when $\ell_1 = N$, \ARC\ will
evict the $LRU(L_1)$. Since by Lemma \ref{lem:a}, $\ell_1' < N$, we
know that for this case, $b_1'$ remains unchanged at 0 and
$\Delta t_1' = +1$. Thus, 
\begin{eqnarray*}
  \Delta \Phi & \leq & - (2 \cdot \Delta t_1') + \Delta \Phi_R \\
  & \leq & -1
\end{eqnarray*}

On the other hand, if $\ell_1 < N$, then \ARC\ will evict the
$LRU(L_2)$. Again, if the cache is full (i.e., $t_1 + t_2 = N$ and
$\ell_1 + \ell_2 = 2N$), then we know that $\ell_2 > N$, which means
that $L_2' \ne L_2$ and $LRU(L_2)$ is not in $L_2'$. Thus, deletion of
$LRU(L_2) = LRU(B_2)$ will not affect $b_2'$ or any of the other
quantities in the potential function. Then comes the {\sc Replace}
step, for which a bound has been proved earlier. Finally, a new item
is brought in and placed in $MRU(T_1)$. Thus $\Delta t_1' \le
1$. Putting it all together, we have
\begin{eqnarray*}
  \Delta \Phi & \leq & - (2 \cdot \Delta t_1') + \Delta \Phi_R \\
  & \leq & -1
\end{eqnarray*}

\paragraph*{Wrapping up the proof of Theorem \ref{thm:1}}

Tying it all up, we have shown that inequality (\ref{eq:2arc}) holds
for every request made after the cache is full, i.e., 
\begin{eqnarray*}
    C_A ( \sigma) + \Delta \Phi \leq 4N \cdot C_O ( \sigma).
\end{eqnarray*}
If we assume that the caches started empty, then the initial potential
is 0, while the final potential can be at most $4N$. Thus, we have 
\begin{eqnarray*}
    C_A ( \sigma) \leq 4N \cdot C_O ( \sigma) + 4N, 
\end{eqnarray*}
thus proving Theorem \ref{thm:1}. \qed

\subsection{Analyzing the Competitiveness of CAR}

Next, we analyze the competitiveness of \CAR. The main result of this
section is the following:

\begin{theorem}
  \label{thm:car1}
  Algorithm \CAR\ is $18N$-competitive.
\end{theorem}

\begin{proof}
  Let $P_X[q]$ be the position of page $q$ in an arbitrary ordered
  sequence of pages $X$. When the set is obvious, we will drop the
  subscript and denote $P_X[q]$ simply by $P[q]$. The set of history
  pages $B_1$ and $B_2$ will be treated as an ordered sequence of
  pages ordered from its LRU position to its MRU position. The set of
  main pages $T_1^0$ (resp., $T_2^0$, $T_1^1$, and $T_2^1$) will be
  treated as an ordered sequence of unmarked (resp., unmarked, marked,
  and marked) pages in $T_1$ (resp, $T_2$, $T_1$, and $T_2$) ordered
  from head to tail. Let $\OPT$ and $\CAR$ be the set of (main and
  history) pages stored in the caches for algorithms \OPT\ and \CAR\
  respectively. Let 
  $D = (T_1 \cup T_2 \cup B_1 \cup B_2) \setminus \OPT$. Thus $D$
  consists of pages in $\CAR$ but not in $\OPT$.
  %
  
  We associate each page with a rank value $R[q]$, which is defined as
  follows:
  \begin{empheq}[box={\mybluebox[5pt]}]{equation}
    R[q] = 
    \begin{cases} 
      P_{B_1}[q]  & \text{if } q \in B_1  \\
      P_{B_2}[q]  & \text{if } q \in B_2  \\
      2 P_{T_1^0}[q] + b_1  &\text{if } q \in T_1^0 \\
      2 P_{T_2^0}[q] + b_2  &\text{if } q \in T_2^0 \\
      3 N + 2 P_{T_1^1}[q] + b_1 &\text{if } q \in T_1^1\\
      3 N + 2 P_{T_2^1}[q] + b_2 &\text{if } q \in T_2^1
    \end{cases}
    \label{eq:R}
  \end{empheq}
  Finally, we define the potential function as follows:
  \begin{empheq}[box={\mybluebox[5pt]}]{equation}
    \Phi = \big(\frac{1}{N-N_O+1}\big) (p +2(b_1+t_1) 
    + 3\sum_{q \in D} R[q])
    \label{eq:Phi}
  \end{empheq}
  The initial value of $\Phi$ is 0. If the following inequality
  (\ref{eq:2}) is true for any request $\sigma$, where $\Delta \Phi$
  is the change in potential caused by serving the request, then when
  summed over all requests, it proves Theorem \ref{thm:car1}.
\begin{equation}
  C_{\CAR}(\sigma) + \Delta \Phi \leq \big(\frac{18N}{N-N_O+1}\big) C_{\OPT}(\sigma). 
  \label{eq:2}
\end{equation}
As before, we assume that request $\sigma$ is processed in two
distinct steps: first when \OPT \: serves and, next when \CAR \:
serves. We will show that inequality (\ref{eq:2}) is satisfied for
each of the two steps.

\subsubsection*{Step 1: \OPT \: serves request $\sigma$}

Since only \OPT \: acts in this step, $C_{\CAR} = 0$, and
$T_1 \cup T_2$ does not change. There are two possible cases: either
\OPT \: faults on $\sigma$ or it does not. If \OPT \: does not fault
on this request, then it is easy to see that $C_{\OPT} = 0$ and
$\Delta \Phi = 0$, thus satisfying inequality (\ref{eq:2}).

If \OPT \: faults on request $\sigma$, then $C_{\OPT} = 1$ and some
page, $q$, is evicted from the cache maintained by \OPT. If $q$ is
maintained by \CAR \: then it follows that $q$ will belong to $D$
after this step and thus its rank will contribute to the potential
function, which will increase by three times the rank of $q$. The
maximal positive change in potential will occur when $q$ is the marked
head page in $T_2$. In this case the rank of $q$ is given by:
$R[q] =3N + 2P[q] + b_2$. The maximal possible values for each of
the terms $P[q]$ and $b_2$ will be $N$, hence the maximum possible
rank of $q$ will be $3N + 2N + N = 6N$. Therefore resulting
potential change is at most $3(6N) = 18N$. 

\subsubsection*{Step 2: \CAR \: serves request $\sigma$}

We break down the analysis into four cases. Case 2.1 deals with the case
when \CAR \: finds the page in its cache. The other three cases assume
that \CAR \: faults on this request because the item is not in
$T_1 \cup T_2$. Cases 2.2 and 2.3 assume that the missing page is found
recorded in the history in lists $B_1$ and $B_2$, respectively. Case 2.4
assumes that the missing page is not recorded in history.

\paragraph*{Case 2.1: \CAR \: has a page hit}

Clearly, the page was found in $T_1 \cup T_2$, and $C_{\CAR} = 0$. We
consider the change of each of terms in the potential function
individually. 

\begin{enumerate}
\item As per the algorithm, $p$ can only change when the page is found
  in history. (See lines 14 through 20 of \CAR$(x)$.) Since the page
  is not found in \CAR's history, $\Delta p = 0$.
\item Neither the cache nor the history lists maintained by \CAR \:
  will change. Thus, the contribution to the second term in $\Phi$,
  i.e., $2(b_1+t_1)$ does not change.
\item Since \OPT \: has already served the page, the page is in \OPT's
  cache. Therefore, even if the page gets marked during this hit, its
  rank value does not change. Thus, the contribution to the last term in $\Phi$,
  also remains unchanged.
\end{enumerate}
We, therefore, conclude that $ \Delta \Phi = 0$, satisfying inequality
(\ref{eq:2}).

Next we will analyze the three cases when the requested page is not in
\CAR's cache. Since $C_{\CAR}= 1$, the change in potential must be at
most $-1$ in each case in order for inequality (\ref{eq:2}) to be
satisfied. Before tackling the three cases, the following lemmas
(\ref{lem:car2.3} and \ref{lem:car2.3r}) are useful for understanding
the potential change caused by the last term in the potential
function, i.e., $\sum_{q \in D} R[q]$. It is worth pointing out that a
call to {\sc Replace} moves either an item from $T_1$ to $B_1$ or from
$T_2$ to $B_2$, which is exactly the premise of Lemma \ref{lem:car2.3}
below. 



\begin{lemma} 
  \label{lem:car2.3}
  When a page is moved from $T_1$ to $B_1$ (or from $T_2$ to $B_2$)
  its rank decreases by at least $1$.
\end{lemma}

\begin{proof} 
  Let $q$ be any page in $T_1$. In order for $q$ to be moved from
  $T_1$ to $B_1$ it must have been unmarked and located at the head of
  $T_1$. Since $P_{T_1}[q] = 1$, the rank of $q$ prior to the move
  must have been $R[q] = 2 P_{T_1}[q] + b_1 = b_1 + 2$, where $b_1$
  is the size of $B_1$ prior to moving $q$.

  After $q$ is moved to the MRU position of $B_1$,
  $R[q] = P_{B_1}[q] = b_1 + 1$. Thus its rank decreased by 1. The
  arguments for the move from $T_2$ to $B_2$ are identical with the
  appropriate changes in subscripts. 
\end{proof}

\begin{lemma} 
  \label{lem:car2.3r}
  When \CAR \: has a page miss, the term $\sum_{q \in D} R[q]$ in the
  potential function $\Phi$ cannot increase.
\end{lemma}

\begin{proof} 
  We examine the rank change based on the original location of the
  page(s) whose ranks changed and in each case show that the rank
  change is never positive. Wherever appropriate we have provided
  references to line numbers in Pseudocode \CAR$(x)$ from Appendix. 

  \begin{description}
  \item[Case A: $q\in B_1 \cup B_2$] {\quad} \\
    The rank of $q\in B_1$, which is simply its position in $B_1$,
    can change in one of three different ways. 
    \begin{enumerate} \topsep=0pt
    \item Some page $x$ less recently used than $q$ (i.e.,
      $P_{B_1}[x] < P_{B_1}[q]$) was evicted (Line 7). In this case,
      it is clear that $P_{B_1}[q]$ decreases by at least 1.
    \item The page $q$ is the requested page and is moved to $T_2$
      (Line 16).  In this case, $q\in\OPT$ and hence its rank cannot
      affect the potential function.
    \item Some page $x$ is added to MRU of $B_1$ (Line 27).  Since
      pages are ordered from LRU to MRU, the added page cannot affect
      the rank of $q$.
    \end{enumerate}
    Using identical arguments for $q\in B_2$, we conclude that a miss
    will not increase the rank of any page in $B_1\cup B_2$.

  \item[Case B: $q\in T_1^0 \cup T_2^0$] {\quad} \\
    The rank of page $q \in T_1^0$, defined as
    $R[q] = 2P_{T_1^0}[q] + b_1$, may be affected in four different
    ways.
    \begin{enumerate} \topsep=0pt
    \item If page $q$ is the head of $T_1$ and gets moved to $B_1$
      (Line 27), by lemma \ref{lem:car2.3}, the change in rank of $q$
      is at most $-1$. 
    \item If an unmarked page $x$ is added to the tail of $T_1$ (Line
      13), then since the ordering is from head to tail, it does not
      affect the position of page $q$. Since there was no change in
      $b_1$, it is clear that the change in $R[q]$ is 0.
    \item If the unmarked page $x \ne q$ at the head of $T_1$ is
      marked and moved to tail of $T_2$ (Line 29), then $P[q]$
      decreases by at least 1. Since the content of $B_1$ is
      unchanged, the change in $R[q] = 2P[q] + b_1$ is at most -2.
    \item If the unmarked page $x \ne q$ at the head of $T_1$ is
      moved to $B_1$ (Line 29), then $P[q]$ decreases by at least 1,
      and $b_1$ increases by $1$. Hence the change in
      $R[q] = 2P[q] + b_1$ is at most -1.
    \end{enumerate}
    The arguments are identical for $q\in T_2^0$. In each case, we have
    shown that a miss will not increase the rank of any page in
    $T_1^0 \cup T_2^0$.

  \item[Case C: $q\in T_1^1$] {\quad} \\
    The rank of page $q\in T_1^1$, defined as $R[q] = 3N +
    2P_{T_1^1}[q] + b_1$, may be affected in four different ways.
    \begin{enumerate} \topsep=0pt
    \item If an unmarked page $x$ is added to the tail of $T_1$ (Line
      13), then since the ordering is from head to tail, it does not
      affect the position of page $q$. Since there was no change in
      $b_1$, it is clear that the change in $R[q]$ is 0.
    \item If the unmarked page $x \ne q$ at the head of $T_1$ is
      marked and moved to tail of $T_2$ (Line 29), then $P[q]$
      decreases by at least 1. Since $B_1$ is unchanged, the change in
      $R[q] = 3N + 2P[q] + b_1$ is at most -2. 
    \item If the unmarked page $x \ne q$ at the head of $T_1$ is
      moved to $B_1$ (Line 29), then $P[q]$ decreases by at least 1,
      and $b_1$ increases by $1$. Hence the change in
      $R[q] = 3N + 2P[q] + b_1$ is at most -1.
    \item Next, we consider the case when the marked page $q$ is the
      head of $T_1$ and gets unmarked and moved to $T_2$ (Line 29). 
      Prior to the move, the rank of $q$ is given by
      $R[q] = 3N + 2P_{T_1^1}[q] + b_1$. Since $B_1$ could be empty, we know
      that $R[q] \geq 3N + 2$. After page $q$ is unmarked and moved
      to $T_2$, its rank is given by $R[q] = 2P_{T_2^0}[q] + b_2$.
      Since $P[q] \le N$ and $b_2 \le N$, we know that the new $R[q]
      \le 3N$. Thus, the rank of page $q$ does not increase. 
    \end{enumerate}
    In each case, we have shown that a miss will not increase the
    rank of any page in $T_1^1$.

  \item[Case D: $q\in T_2^1$] {\quad} \\
    The rank of page $q\in T_2^1$, defined as $R[q] = 3N +
    2P_{T_2^1}[q] + b_2$, may be affected in four different ways.
    \begin{enumerate}
    \item If an unmarked page $x$ is added to the tail of $T_2$ (Lines
      16, 19, or 29), and if $b_2$ does not change, it is once again
      clear that the change in $R[q]$ is 0.    
    \item If a marked page $x \ne q$ at the head of $T_2$ gets unmarked and
      moved to the tail of $T_2$ (Line 36), the position of $q$ will
      decrease by 1 and there is no change in $b_2$. Thus $R[q]$
      changes by at most -2. 
    \item If an unmarked page $x$ at the head of $T_2$ is moved to
      $B_2$ (Line 34), $P[q]$ decreases by 1 and $b_2$ increases by
      1. Thus $R[q]$ changes by at most -1.
    \item Finally, we consider the case when the marked page $q$ is
      the head of $T_2$ and gets unmarked and moved to the tail of
      $T_2$ (Line 36). Prior to the move, the rank of $q$ is given by
      $R[q] = 3N + 2P_{T_2^1}[q] + b_2$. Even if $B_2$ is empty, we know
      that $R[q] \geq 3N + 2$. After page $q$ is unmarked and moved
      to $T_2$, its rank is given by $R[q] = 2P_{T_2^0}[q] + b_2$.
      Since $P[q] \le N$ and $b_2 \le N$, we know that the new $R[q]
      \le 3N$. Thus, the rank of page $q$ does not increase. 
    \end{enumerate}
    In each case, we have shown that a miss will not increase the
    rank of any page in $T_2^1$. 
  \end{description}
  The four cases (A through D) together complete the proof of Lemma
  \ref{lem:car2.3r}.
\end{proof}

\noindent We continue with the remaining cases for the proof of Theorem
\ref{thm:car1}. 
\paragraph*{Case 2.2: \CAR \: has a page miss and the missing page
  is in $B_1$} 
We consider the change in the potential function (defined in
Eq. \ref{eq:Phi}) by analyzing each of its three terms.
\begin{enumerate} \topsep=0pt\itemsep=0pt
\item Value of $p$ increases by 1, except when it is equal to 
  $N$, in which case it remains unchanged. (See Line 15.) Thus, the
  first term increases by at most 1. 
\item 
  The call to {\sc Replace} has no effect on the value of
  $(t_1 + b_1)$ because an item is moved either from $T_1$ to $B_1$ or
  from $T_2$ to $B_2$. Since the requested page in $B_1$ is moved to
  $T_2$, $(t_1 + b_1)$ decreases by 1.
\item By Lemma \ref{lem:car2.3r}, we already know that the last term
  increases by at most 0. 
\end{enumerate}
Since $p$ increases by at most 1 and the term $2(t_1 + b_1)$ decreases
by at least 2, the total change in the potential function, is at most
-1.

\paragraph*{Case 2.3: \CAR \: has a page miss and the missing page
  is in $B_2$}

When the missing page is in $B_2$, \CAR\ makes a call to {\sc Replace}
(Line 5) and then executes Lines 18-19. Thus, $p$ is decremented
except if it is already equal to 0.
We consider two subcases: $\Delta p < 0$ and $\Delta p = 0$.
\paragraph*{\underline{$\Delta p < 0$}: }
As in Case 2.2, the call to {\sc Replace} has no effect on
$(t_1 + b_1)$. Since, Lines 18-19 do not affect $T_1\cup B_1$, the
second term does not change. By Lemma \ref{lem:car2.3r}, we know that
the last term increases by at most 0.
Since $\Delta p \le -1$, the total change in the potential function,
$\Delta p + \Delta 2(t_1 + b_1)$ is at most -1.

\paragraph*{\underline{$\Delta p = 0$}: }
Unlike the subcase above when $p$ decreases by 1, the change in $p$
cannot guarantee the required reduction in the potential. We therefore
need a tighter argument. 
We know that there is a call to {\sc Replace}. Three cases
arise and are discussed below. 
\begin{itemize} \topsep=0pt\itemsep=0pt
\item If $T_1$ is empty, then $T_2$ must have $N$ pages, at least one
  of which must be in $D$. Also, {\sc Replace} must act on $T_2$,
  eventually evicting an unmarked page from head of $T_2$, causing the
  rank of any page from $T_2 \setminus \OPT$ to decrease by 1.
\item If $T_1$ is not empty and has at least one page from $D$, then
  the condition in Line 24 passes and {\sc Replace} must act on $T_1$,
  eventually evicting an unmarked page from head of $T_1$, causing the
  rank of at least one page from $T_1 \setminus \OPT$ to decrease by 1.
\item Finally, if $T_1$ is not empty and all its pages are in \OPT,
  then $T_2$ must have a page $q \in D$. Since the requested page $x$
  was found in $B_2$ and is moved to the tail of $T_2$, even though
  the position of $q$ in $T_2$ does not change, $b_2$ decreased by 1
  and consequently the rank of $q$ decreases by 1.
\end{itemize}
Thus, in each case, even though neither $p$ nor the quantity
$(t_1 + b_1)$ changed, the third term involving ranks, and
consequently, the potential function decreased by at least 3.

The following two lemmas are useful for Case 2.4, when the missing
page is not in $T_1\cup T_2\cup B_1\cup B_2$.
\begin{lemma} 
  \label{lem:LeqN}
  We make two claims: 
  \begin{enumerate} \topsep=0pt\itemsep=0pt
  \item If $t_1 + b_1 = N$ and the $LRU$ page of $B_1$ is evicted from
    the cache on Line 7, then $\sum_{q \in D} R[q]$ will decrease by
    at least one.
    \item If $t_2 + b_2 > N$, and the $LRU$ page of $B_2$, is evicted from
    the cache on Line 9, then $\sum_{q \in D} R[q]$ will decrease by
    at least one.
  \end{enumerate}
\end{lemma}

\begin{proof} 
  We tacke the first claim. Assume that $y$ is the $LRU$ page of $B_1$
  that is being evicted on Line 7. Then \CAR\ must have had a page
  miss on $x \not\in B_1 \cup B_2$, and the requested page $x$ is
  added to the tail of $T_1$.
  Since $t_1 + b_1 = N$, there is at least one page
  $q \in T_1 \cup B_1$ that is not in \OPT's cache and whose rank
  contributes to the potential function. First, we assume that
  $q \in T_1\setminus\OPT$, whose rank is given by:
  $R[q] = 2 *P[q] + b_1$. For each of the three cases, we show that
  the potential function does decrease by at least 1.
  \begin{itemize} \topsep=0pt\itemsep=0pt
  \item If {\sc Replace} acts on $T_1$ and the unmarked head of $T_1$,
    different from $q$, is moved to $B_1$ then the size of $B_1$
    remains the same (because a page gets added to $B_1$ while another
    page is evicted) but the position of $q$ in $T_1$ decreases by
    one. Therefore $R[q]$ decreases by 2.
  \item If {\sc Replace} acts on $T_1$ and $q$ itself is moved to
    $B_1$ then by Lemma \ref{lem:car2.3}, $R[q]$ decreases by at least
    1.
  \item If {\sc Replace} acts on $T_2$, then we use the fact that a
    page is evicted from $B_1$, and the $b_1$ term in $R[q]$ must
    decrease by 1.
  \end{itemize}
  %
  %
  Next, we assume that $q \in B_1\setminus\OPT$. Since $LRU(B_1)$ is
  evicted, the position of the page $q$ will decrease by one. Thus
  $R[q] = P_{B_1}[q]$ must decrease by at least 1, completing the proof of the
  first claim in the lemma. 
  
  The proof of the second claim is very similar and only requires
  appropriate changes to the subscripts.
\end{proof}

Next we tackle the last case in the proof of Theorem \ref{thm:car1}.
\paragraph*{Case 2.4: \CAR \: has a page miss and the missing page
  is not in $B_1 \cup B_2$}

We assume that \CAR's cache is full (i.e., $l_1 + l_2 = 2N$).  We
consider two cases below -- first, if $l_1 = N$ and the next when $l_1
< N$. 

If $l_1 = t_1 + b_1 = N$, \CAR \: will call {\sc Replace}, evict
$LRU(B_1)$ and then add the requested page to the tail of
$T_1$. Below, we analyze the changes to the three terms in the
potential function.
\begin{itemize} \topsep=0pt\itemsep=0pt
\item Since $p$ is not affected, the first term does not change. 
\item Since a page is added to $T_1$ and a page is evicted from $B_1$,
  the net change in the second term is 0.
\item Since the conditions of Lemma \ref{lem:LeqN} apply, the total
  rank will decrease by at least 1. 
\end{itemize}
Adding up all the changes, we conclude that the potential function
decreases by at least 3. 

If $l_1 < N$, \CAR \: will call {\sc Replace}, evict $LRU(B_2)$
and then add a page to the tail of $T_1$. As above, we analyze the
changes to the three terms in the potential function.
\begin{itemize} \topsep=0pt\itemsep=0pt
\item Since $p$ is not affected, the first term does not change. 
\item A page is added to $T_1$ and a page is evicted from $B_2$ hence
  $(t_1+b_1)$ increases by 1.
\item Since $l_2 > N$, the conditions of Lemma \ref{lem:LeqN} apply,
  the total rank will decrease by at least 1.
\end{itemize}

Adding up all the changes, we conclude that the potential function
decreases by at least 1, thus completing Case 2.4.

\paragraph*{Wrapping up the proof of Theorem \ref{thm:car1}: }
Combining the four cases (2.1 through 2.4) proves that inequality
(\ref{eq:2}) is satisfied when \CAR\ serves request $\sigma$. This
completes the proof of Theorem \ref{thm:car1}, establishing that the
upper bound on the competitiveness of \CAR\ is $18N$.
\end{proof} 

\section{Conclusions and Future Work}
\label{sec:concl}

Adaptive algorithms are tremendously important in situations
where inputs are infinite online sequences and no single optimal
algorithm exists for all inputs. Thus, different portions of the input
sequence require different algorithms to provide optimal
responses. Consequently, it is incumbent upon the algorithm to sense
changes in the nature of the input sequence and adapt to these
changes. Unfortunately, these algorithms are harder to analyze. We
present the analysis of two important adaptive algorithms called \ARC\
and \CAR\ and show that they are competitive along with proving good
lower bounds on the competitiveness ratios. 

Two important open questions remain unanswered. Given that there is a
gap between the lower and upper bounds on the competitiveness ratios
of the two adaptive algorithms, \ARC\ and \CAR, what is the true
ratio? More importantly, is there an ``expected'' competitiveness
ratio for request sequences that come from real applications? The
second question would help explain why \ARC\ and \CAR\ perform better
in practice than \LRU\ and \CLOCK, respectively.  

\paragraph*{Acknowledgments} This work was partly supported by two
NSF Grants (CNS-1018262 and CNS-1563883) and the NSF Graduate Research
Fellowship (DGE-1038321). We are grateful to Kirk Pruhs for suggesting
enhancing our results with the assumption of unequal cache sizes. 

\vfill\pagebreak

\bibliography{paging}{}

\newcommand{\etalchar}[1]{$^{#1}$}
\begin{thebibliography}{LCK{\etalchar{+}}01}

\bibitem[Alb96]{a-coa-96}
S.~Albers.
\newblock Competitive online algorithms.
\newblock Technical report, BRICS Lecture Series, Computer Science Department,
  University of Aarhus, 1996.

\bibitem[BM04]{bm-CAR-04}
S.~Bansal and D.~S. Modha.
\newblock {CAR}: {CLOCK} with adaptive replacement.
\newblock In {\em Proceedings of the 3rd USENIX Conference on File and Storage
  Technologies}, FAST '04, pages 187--200, Berkeley, CA, USA, 2004. USENIX
  Association.

\bibitem[Cor68]{c-pems-68}
F.~J. Corbato.
\newblock A paging experiment with the {MULTICS} system.
\newblock Technical report, DTIC Document, 1968.

\bibitem[Fri99]{f-wntprp-99}
M.~B. Friedman.
\newblock Windows {NT} page replacement policies.
\newblock In {\em Proceedings of the Intl. CMG Conference}, pages 234--244,
  1999.

\bibitem[Hoc97]{h-aanp-97}
D.~S. Hochbaum, editor.
\newblock {\em Approximation algorithms for {NP}-hard problems}.
\newblock PWS Publishing Co., Boston, MA, USA, 1997.

\bibitem[JCZ05]{jcz-cpro-05}
S.~Jiang, F.~Chen, and X.~Zhang.
\newblock {CLOCK-Pro}: An effective improvement of the {CLOCK} replacement.
\newblock In {\em USENIX Annual Technical Conference, General Track}, pages
  323--336, 2005.

\bibitem[JIPP10]{jipp-dcacr-10}
A.~Janapsatya, A.~Ignjatovic, J.~Peddersen, and S.~Parameswaran.
\newblock Dueling {CLOCK}: adaptive cache replacement policy based on the
  {CLOCK} algorithm.
\newblock In {\em Design, Automation \& Test in Europe Conference \& Exhibition
  (DATE), 2010}, pages 920--925. IEEE, 2010.

\bibitem[JS94]{JS94}
T.~Johnson and D.~Shasha.
\newblock {2Q}: A low overhead high performance buffer management replacement
  algorithm.
\newblock In {\em Proc. of VLDB}, pages 297--306, 1994.

\bibitem[JZ02]{jz-LIRS-02}
S.~Jiang and X.~Zhang.
\newblock {LIRS}: An efficient low inter-reference recency set replacement
  policy to improve buffer cache performance.
\newblock In {\em Proc. ACM Sigmetrics Conf.}, pages 297--306. ACM Press, 2002.

\bibitem[LCK{\etalchar{+}}01]{lcknmck-LRFU-01}
D.~Lee, J.~Choi, J.~H. Kim, S.~H. Noh, S.~L. Min, Y.~Cho, and C.~S. Kim.
\newblock {LRFU}: A spectrum of policies that subsumes the least recently used
  and least frequently used policies.
\newblock {\em IEEE Trans. Comput.}, 50(12):1352--1361, December 2001.

\bibitem[MM03]{mm-ARC-03}
N.~Megiddo and D.~S. Modha.
\newblock {ARC}: A self-tuning, low overhead replacement cache.
\newblock In {\em Proceedings of the 2nd USENIX Conference on File and Storage
  Technologies}, FAST '03, pages 115--130, Berkeley, CA, USA, 2003. USENIX
  Association.

\bibitem[MM04]{mm-ARC-04}
N.~Megiddo and D.~S. Modha.
\newblock Outperforming {LRU} with an adaptive replacement cache algorithm.
\newblock {\em {IEEE} Computer}, 37(4):58--65, 2004.

\bibitem[OOW93]{oow-LRUK-93}
E.~J. O'Neil, P.~E. O'Neil, and G.~Weikum.
\newblock The {LRU-K} page replacement algorithm for database disk buffering.
\newblock {\em SIGMOD Rec.}, 22(2):297--306, June 1993.

\bibitem[ST85]{st-aelupr-85}
D.~D. Sleator and R.~E. Tarjan.
\newblock Amortized efficiency of list update and paging rules.
\newblock {\em Commun. ACM}, 28(2):202--208, February 1985.

\end{thebibliography}

\vfill\pagebreak
\section{Appendix}
\label{sec:App}

We reproduce the pseudocode for \ARC\ and \CAR\ below. 
\begin{algorithm}

{\bf Pseudocode}: \ARC($x$) \\
{\sc INPUT}: The requested page $x$ \\
{\sc INITIALIZATION}: Set $p=0$ and set lists $T_1$, $B_1$, $T_2$, and $B_2$ to empty\\
\label{alg:ARC}
\begin{algorithmic}[1] 
\smallskip
\If{($x$ is in $T_1 \cup T_2$)} 
\Comment{cache hit}  
\State Move $x$ to the top of $T_2$
\ElsIf {($x$ is in $B_1$)}
\Comment{cache history hit}  
\State {\sc Adaptation}: Update $p =\min\{p+1,N\}$
\Comment{learning rate = 1}
\State {\sc Replace}()
\Comment {make space in $T_1$ or $T_2$}
\State Fetch $x$ and move to the top of $T_2$
\ElsIf {($x$ is in $B_2$)}
\Comment{cache history hit}  
\State {\sc Adaptation}: Update: $p=\max\{p- 1,0\}$
\Comment{learning rate = 1}
\State {\sc Replace}()
\Comment {make space in $T_1$ or $T_2$}
\State Fetch $x$ and move to the top of $T_2$ 
\Else \Comment{cache and history miss}
\If {($t_1 + b_1 = N$)}
\If {($t_1 < N$)}
\State Discard LRU item in $B_1$
\State {\sc Replace}()
\Comment {make space in $T_1$ or $T_2$}
\Else
\State Discard LRU page in $T_1$ and remove from cache
\EndIf
\ElsIf {$((t_1 + b_1 < N$) and ($t_1 + t_2 + b_1 + b_2 \ge N$))}
\If {($t_1 + t_2 + b_1 + b_2 = 2N$)}
\State Discard LRU item in $B_2$
\EndIf
\State {\sc Replace}()
\Comment {make space in $T_1$ or $T_2$}
\EndIf
\State Fetch $x$ and move to the top of $T_1$ 
\EndIf
\end{algorithmic}
\hrulefill

{\sc Replace}()
\begin{algorithmic}[1]
\makeatletter\setcounter{ALG@line}{25}\makeatother
\If{(($t_1 \ge 1$) and (($x \in B_2$ and $t_1 = p$) or ($t_1 > p$)))}
\State Discard LRU page in $T_1$ and insert as MRU history item in $B_1$
\Else
\State Discard LRU page in $T_2$ and insert as MRU history item in $B_2$
\EndIf
\end{algorithmic}
\end{algorithm}

\begin{algorithm}
{\bf Pseudocode}: \CAR($x$) \\
{\sc INPUT}: The requested page $x$ \\
{\sc INITIALIZATION}: Set $p=0$ and set lists $T_1$, $B_1$, $T_2$, and $B_2$ to empty\\
\label{alg:CAR}
\begin{algorithmic}[1] 
\smallskip
\If{($x$ is in $T_1 \cup T_2$)} 
\Comment{cache hit}  
\State Mark page $x$
\Else \Comment{cache miss}
\If{$(t_1 + t_2 = N)$} 
\Comment{cache full, replace a page from cache}
\State {\sc Replace}()
\Comment {make space in $T_1$ or $T_2$}
\If {(($x \not\in B_1 \cup B_2$) and ($t_1 + b_1 = N$))}
\State Discard LRU page in $B_1$
\ElsIf {(($x \not\in B_1 \cup B_2$) and ($t_1 + t_2 + b_1 + b_2 = 2N$))}
\State Discard LRU page in $B_2$.
\EndIf
\EndIf
\If{($x \not\in B_1 \cup B_2)$}
\Comment{cache miss}
\State Insert $x$ at the tail of $T_1$; Unmark page $x$
\ElsIf{($x \in B_1$)}
\Comment{cache history hit}
\State {\sc Adaptation}: Update $p =\min\{p+1,N\}$
\Comment{learning rate = 1}
\State Move $x$ to the tail of $T_2$; Unmark page $x$
\Else
\Comment{cache history hit}
\State {\sc Adaptation}: Update: $p=\max\{p-1,0\}$
\Comment{learning rate = 1}
\State Move $x$ to the tail of $T_2$; Unmark page $x$
\EndIf
\EndIf
\end{algorithmic}
\hrulefill

{\sc Replace}()
\begin{algorithmic}[1]
\makeatletter\setcounter{ALG@line}{21}\makeatother
\State found = false
\Repeat
\If{($t_1 \ge \max\{1,p\}$)}
\If{(head page in $T_1$ is unmarked)}
\State found = true
\State Discard head page in $T_1$ and insert as MRU history item in $B_1$
\Else
\State Unmark head page in $T_1$, move page as tail page in $T_2$, and
move head of $T_1$ clockwise
\EndIf
\Else
\If{(head page in $T_2$ is unmarked)}
\State found = true
\State Discard head page in $T_2$ and insert as MRU history item in $B_2$
\Else
\State Unmark head page in $T_2$, and move head of $T_2$ clockwise
\EndIf
\EndIf
\Until{(found)}
\end{algorithmic}
\end{algorithm}

\end{document}